\documentclass[11pt,a4paper]{article}
\pdfoutput=1

\usepackage[compress]{cite}
\usepackage{fullpage}
\usepackage{microtype}
\usepackage{amsthm}
\usepackage{amsmath}
\usepackage{cmll}
\usepackage{multicol}
\usepackage[T1]{fontenc}
\usepackage[utf8]{inputenc}
\usepackage{fontawesome}
\usepackage{graphicx}
\usepackage{proof}
\usepackage{float}
\usepackage{amssymb}
\usepackage{stmaryrd}
\usepackage{hyperref}
\hypersetup{
  colorlinks=true,
  linkcolor=blue,
  anchorcolor=blue,
  citecolor=blue,
  urlcolor=blue,
  filecolor=blue,
  breaklinks=true
}

\floatstyle{boxed}
\restylefloat{figure}

\newcommand\zero{\ensuremath{\mathfrak 0}}
\newcommand\one{\ensuremath{\mathfrak 1}}
\newcommand\ket[1]{\ensuremath{|{#1}\rangle}}
\newcommand\abstr[1]{#1.}
\newcommand\B{{\mathcal B}}
\newcommand\Q{{\mathcal Q}}
\newcommand\inl{\mbox{\it inl}}
\newcommand\inr{\mbox{\it inr}}
\newcommand\elimone{\delta_{\one}}
\newcommand\elimzero{\delta_{\zero}}
\newcommand\elimwith{\delta_{\with}}
\newcommand\elimplus{\delta_{\oplus}}
\newcommand\elimsup{\delta_{\odot}}
\newcommand\elimtens{\delta_{\otimes}}
\newcommand\test{\mbox{\it if}}
\newcommand\plus{\mathbin{\text{\normalfont\scalebox{0.7}{\faPlus}}}}
\newcommand\length[1]{\ell(#1)}

\newcommand\irule[3]{\infer[\mbox{\footnotesize $#3$}]{#2}{#1}}

\newcommand\pair[2]{\langle #1, #2 \rangle}
\newcommand\topintro{\langle\rangle}
\newcommand\super[2]{[#1,#2]}

\newcommand\lra{\longrightarrow}
\newcommand\lla{\longleftarrow}
\newcommand\lras{\lra^*}
\newcommand\llas{\mathrel{{}^*{\longleftarrow}}}

\newtheorem{lemma}{Lemma}[section]
\newtheorem{theorem}[lemma]{Theorem}
\newtheorem{definition}[lemma]{Definition}
\newtheorem{corollary}[lemma]{Corollary}
\newtheorem{remark}[lemma]{Remark}
\newtheorem{example}[lemma]{Example}

\floatstyle{boxed}
\restylefloat{figure}

\title{A linear linear lambda-calculus\footnote{Partially supported by the French-Argentinian IRP SINFIN.}}

\author{
  Alejandro D\'iaz-Caro\\
  {\footnotesize DCyT, Universidad Nacional de Quilmes, Argentina}\\
  {\footnotesize ICC (CONICET--Universidad de Buenos Aires), Argentina}\\
  {\footnotesize DC--FCEN, Universidad de Buenos Aires, Argentina}\\
  \texttt{\footnotesize adiazcaro@conicet.gov.ar}
  \and
  {Gilles Dowek}\\
  {\footnotesize Inria \& ENS Paris-Saclay, France}\\
  \texttt{\footnotesize gilles.dowek@ens-paris-saclay.fr}
}
\date{}

\begin{document}

\maketitle

\begin{abstract}
  We present a linearity theorem for a proof language of intuitionistic
  multiplicative additive linear logic, incorporating addition and scalar
  multiplication. The proofs in this language are linear in the algebraic
  sense. This work is part of a broader research program aiming to define a
  logic with a proof language that forms a quantum programming language.
\end{abstract}

\section{Introduction}

\subsection{Interstitial rules}

The name of linear logic \cite{Girard87} suggests that this logic has
some relation with the algebraic notion of linearity. A common account
of this relation is that a proof of a linear implication between two
propositions $A$ and $B$ should not be any function mapping proofs of
$A$ to proofs of $B$, but a linear one.  This idea has been fruitfully
exploited to build models of linear logic
(for example \cite{Blute96,Girard99,Ehrhard02}), but it seems
difficult to even formulate it within the proof language
itself. Indeed, expressing the properties $f(u + v) = f(u) + f(v)$ and
$f(a. u) = a. f(u)$ requires an addition and a multiplication by a
scalar, that are usually not present in proof languages.

The situation has changed with quantum programming languages
and the algebraic $\lambda$-calculus,
that mix
some usual constructions of programming languages with algebraic
operations.

In this paper, we construct a minimal extension of the proof language
for intuitionistic multiplicative additive linear logic with addition
and multiplication by a scalar, the ${\mathcal L}^{\mathcal
  S}$-calculus (where $\mathcal S$ denotes the semi-ring of scalars
used), and we prove that the proof language of this logic expresses
linear maps only: if $f$ is a proof of an implication between two
propositions, then $f(u + v) = f(u) + f(v)$ and $f(a. u) = a. f(u)$.

Our main goal is thus to construct this extension of
intuitionistic linear logic and prove this linearity theorem.
Only in a second step, 
we discuss whether such a language forms the basis of a quantum
programming language or not.

In classical linear logic, the right rules of the multiplicative
falsehood, the additive implication, and the multiplicative
disjunction
$$
\irule{\Gamma\vdash\Delta}{\Gamma\vdash\bot,\Delta}{\mbox{$\bot$-r}}
\qquad
\irule{\Gamma, A \vdash \Delta}
  {\Gamma \vdash A \Rightarrow B, \Delta}
  {\mbox{$\Rightarrow$-r1}}
  \qquad
  \irule{\Gamma \vdash B, \Delta}
  {\Gamma \vdash A \Rightarrow B, \Delta}
  {\mbox{$\Rightarrow$-r2}}
\qquad
\irule{\Gamma\vdash A,B,\Delta}{\Gamma\vdash A\parr B,\Delta}{\mbox{$\parr$-r}}
$$
do not preserve the number of propositions in the right-hand side of
the sequents.  Hence, these three connectives are excluded from
intuitionistic linear logic, and we do not consider them.

Thus, we have
the multiplicative truth $\one$,
the multiplicative implication $\multimap$, 
the multiplicative conjunction $\otimes$,
the additive truth $\top$,
the additive falsehood $\zero$,
the additive conjunction $\with$,
and the additive disjunction $\oplus$.

The introduction rule for the additive conjunction $\with$
is the same as that in usual natural deduction
$$\irule{\Gamma \vdash A & \Gamma \vdash B}
{\Gamma \vdash A \with B}
{\mbox{$\with$-i}}$$
In particular, the proofs of $A$, $B$, and $A \with B$ are in the same context
$\Gamma$. In contrast, 
in the introduction rule for the multiplicative conjunction $\otimes$
$$\irule{\Gamma_1 \vdash A & \Gamma_2 \vdash B}
{\Gamma_1, \Gamma_2 \vdash A \otimes B}
{\mbox{$\otimes$-i}}$$
the proofs of $A$ and $B$ are in two contexts $\Gamma_1$ and $\Gamma_2$
and the proof of the conclusion $A \otimes B$ is in the multiset union of
these two contexts.
But, in both cases, in the elimination rules
$$\irule{\Gamma \vdash A \with B & \Delta, A \vdash C}
{\Gamma, \Delta \vdash C}
{\mbox{$\with$-e1}}
\hspace{1cm}
\irule{\Gamma \vdash A \with B & \Delta, B \vdash C}
{\Gamma, \Delta \vdash C}
{\mbox{$\with$-e2}}$$
$$\irule{\Gamma \vdash A \otimes B & \Delta, A, B \vdash C}
{\Gamma, \Delta \vdash C}
{\mbox{$\otimes$-e}}$$
the proof of the major premise
and that of the minor one are in contexts
$\Gamma$ and $\Delta, A$ (resp.~$\Delta, B$, $\Delta, A, B$) and the proof
of the conclusion $C$ is in the multiset union of $\Gamma$ and $\Delta$.
The same holds for the other connectives.

To extend this logic with addition and multiplication by a scalar,
we proceed, like in \cite{odot}, in two steps: we first add
interstitial rules and then scalars.

An interstitial rule is a deduction rule whose premises are identical
to its conclusion. We consider two such rules
$$\irule{\Gamma \vdash A & \Gamma \vdash A}{\Gamma \vdash A}{\mbox{sum}}
\hspace{1cm}
\irule{\Gamma \vdash A}{\Gamma \vdash A}{\mbox{prod}}$$
These rules obviously do not extend provability, but they introduce
new constructors $\plus$ and $\bullet$ in the proof language.

We then consider a semi-ring ${\mathcal S}$ of scalars and replace the
introduction rule of the connective $\one$ with a family of rules
$\one$-i$(a)$, one for each scalar, and the rule prod with a family of
rules prod$(a)$, also one for each scalar
$$\irule{}{\vdash \one}{\mbox{$\one$-i} (a)}
\qquad\qquad\qquad
\irule{\Gamma \vdash A}{\Gamma \vdash A}{\mbox{prod} (a)}$$

\subsection{Commutations}

Adding these rules yields proofs that cannot be reduced,
because the introduction rule of some connective and its elimination
rule are separated by an interstitial rule, for example
$$\irule{\irule{\irule{\irule{\pi_1}{\Gamma \vdash A}{}
      & 
    \irule{\pi_2}{\Gamma \vdash B}{}}
    {\Gamma \vdash A \with B}
    {\mbox{$\with$-i}}
    &
    \irule{\irule{\pi_3}{\Gamma \vdash A}{}
      &
      \irule{\pi_4}{\Gamma \vdash B}{}
    }
    {\Gamma \vdash A \with B}
    {\mbox{$\with$-i}}
  }
  {\Gamma \vdash A \with B}
  {\mbox{sum}}
  &
  \irule{\pi_5}{\Gamma, A \vdash C}{}
}
{\Gamma \vdash C}
{\mbox{$\with$-e1}}$$
Reducing such a proof, sometimes called a {\it commuting cut},
requires reduction rules to commute the rule sum either with the
elimination rule below or with the introduction rules above.

As the commutation with the introduction rules above is not always
possible, for example in the proofs
$$\irule{\irule{\irule{\pi_1}{\Gamma \vdash A}{}}
  {\Gamma \vdash A \oplus B}
  {\mbox{$\oplus$-i1}}
  &
  \irule{\irule{\pi_2}{\Gamma \vdash B}{}}
  {\Gamma \vdash A \oplus B}
  {\mbox{$\oplus$-i2}}
}
{\Gamma \vdash A \oplus B}
{\mbox{sum}}$$
$$\irule{\irule{\irule{\pi_1}{\Gamma_1 \vdash A}{}
    &  \irule{\pi_2}{\Gamma_2 \vdash B}{}
  }
  {\Gamma \vdash A \otimes B}
  {\mbox{$\otimes$-i}}
  &
  \irule{\irule{\pi_3}{\Gamma'_1 \vdash A}{}
    & \irule{\pi_4}{\Gamma'_2 \vdash B}{}
  }
  {\Gamma \vdash A \otimes B}
  {\mbox{$\otimes$-i}}
}
{\Gamma \vdash A \otimes B}
{\mbox{sum}}$$
where $\Gamma_1 \Gamma_2 = \Gamma'_1 \Gamma'_2 = \Gamma$, the
commutation with the elimination rule below is often preferred.  In
this paper, we favour the commutation of the interstitial rules with
the introduction rules, rather than with the elimination rules,
whenever it is possible, that is for the connectives $\one$,
$\multimap$, $\top$, and $\with$, and keep commutation with the
elimination rules for the connectives $\otimes$ and $\oplus$ only.
For example, with the additive conjunction $\with$, the proof
$$\irule{\irule{\irule{\pi_1}{\Gamma \vdash A}{}
    &
  \irule{\pi_2}{\Gamma \vdash B}{}}
  {\Gamma \vdash A \with B}
  {\mbox{$\with$-i}}
  &
  \irule{\irule{\pi_3}{\Gamma \vdash A}{}
    &
    \irule{\pi_4}{\Gamma \vdash B}{}
  }
  {\Gamma \vdash A \with B}
  {\mbox{$\with$-i}}
}
{\Gamma \vdash A \with B}
{\mbox{sum}}$$
reduces to 
$$\irule{\irule{\irule{\pi_1}{\Gamma \vdash A}{}
    &
  \irule{\pi_3}{\Gamma \vdash A}{}}
  {\Gamma \vdash A}
  {\mbox{sum}}
  &
  \irule{\irule{\pi_2}{\Gamma \vdash B}{}
    &
    \irule{\pi_4}{\Gamma \vdash B}{}
  }
  {\Gamma \vdash B}
  {\mbox{sum}}
}
{\Gamma \vdash A \with B}
{\mbox{$\with$-i}}$$
Such commutation rules yield a stronger introduction property for the
considered connective.

For coherence, we commute both rules sum and prod with the elimination
rule of the additive disjunction $\oplus$
and of the multiplicative conjunction
$\otimes$, rather that with its
introduction rules. But, for the rule prod, both choices are
possible.

\subsection{Related work}
While our primary objective is to introduce a minimal extension to the proof
language of linear logic, our work is greatly indebted to quantum programming
languages. These languages were pioneers in amalgamating programming language
constructs with algebraic operations, such as addition and scalar
multiplication.

The language QML~\cite{AltenkirchGrattageLICS05} introduced the
concept of superposition of terms, through an encoding: the
$\textrm{if}^\circ$ constructor can receive qubits as conditional
parameters. For example, the expression $\textrm{if}^\circ\ a.\ket 0 +
b.\ket 1\ \textrm{then}\ u\ \textrm{else}\ v$ represents the linear
combination $a.u+b.v$.  Thus, although QML does not have a direct way
to represent linear combinations of terms, such linear combinations
can always be expressed using this $\textrm{if}^\circ$ constructor.  A
linearity property, and even an unitarity property, is proved for QML,
through a translation to quantum circuits.

The ZX calculus~\cite{ZXBook17} is a graphical language based on a
categorical model. It does not have addition or multiplication by a
scalar in the syntax, but such constructions could be added and
interpreted in the model. This idea of extending the syntax with
addition and multiplication by a scalar, lead to the Many Worlds
Calculus~\cite{Chardonnet23}.
Although the Many Worlds Calculus and the
${\mathcal L}^{\mathcal S}$-calculus have several points in common,
the ${\mathcal L}^{\mathcal S}$-calculus
takes advantage of being a $\lambda$-calculus and not
a graphical language to introduce a primitive $\lambda$-abstraction, while the
Many Worlds Calculus introduces it indirectly through the adjunction between
the hom and the tensor. Then, the linearity proof for the Many
Worlds Calculus uses semantic tools, while that for the ${\mathcal L}^{\mathcal
S}$-calculus is purely syntactic.

The algebraic $\lambda$-calculus~\cite{Vaux2009} and Lineal~\cite{Lineal} have
similar syntaxes to the $\mathcal L^{\mathcal S}$-calculus we are proposing.
However, in the case of the algebraic $\lambda$-calculus, there is only a simple
intuitionistic type system, with no proof of linearity. In the case of Lineal,
there is no type system, and the linearity is not proved, but forced: 
the term $f(u+v)$ is defined as $f(u)+f(v)$ and $f(a.u)$ is defined as $a.f(u)$.
Several type systems have been proposed for
Lineal~\cite{ArrighiDiazcaroLMCS12,DiazcaroPetitWoLLIC12,ArrighiDiazcaroValironIC17,LambdaS,DiazcaroGuillermoMiquelValironLICS19},
but none of them are related to linear logic, and they are not intended to
prove the linearity, instead of forcing it.

Finally, other sources of the 
${\mathcal L}^{\mathcal S}$-calculus are the 
quantum lambda calculus~\cite{SelingerValironMSCS06} and the
language Q~\cite{ZorziMSCS16}, although the classical nature of
their control yields a restricted form of superposition, on
data rather than on arbitrary terms.

\subsection{Outline of the paper}

Extending the proof language of intuitionistic linear logic with
interstitial rules and with scalars yields the ${\mathcal L}^{\mathcal
  S}$-calculus, that we define and study in Section~\ref{secls}.  In
particular, we prove that the ${\mathcal L}^{\mathcal S}$-calculus
verifies the subject reduction, confluence, termination, and
introduction properties.  We then show, in
Section~\ref{secvectorsmatrices}, that the vectors of ${\mathcal S}^n$
can be expressed in this calculus, that the irreducible closed proofs
of some propositions are equipped with a structure of vector space,
and that all linear functions from ${\mathcal S}^m$ to ${\mathcal
  S}^n$ can be expressed as proofs of an implication between such
propositions.  We then prove, in Section~\ref{seclinearity}, the main
result of this paper: that, conversely, all the proofs of implications
are linear.

Finally, we discuss applications to quantum computing, in
Section~\ref{seclss}.

\section{The \texorpdfstring{${\mathcal L}^{\mathcal S}$}{L-S}-calculus}
\label{secls}

\subsection{Syntax and operational semantics}

\begin{figure}[t]
  $$\irule{}
  {x:A \vdash x:A}
  {\mbox{ax}}
  \qquad
  \irule{\Gamma \vdash t:A & \Gamma \vdash u:A}
  {\Gamma \vdash t \plus u:A}
  {\mbox{sum}}
  \qquad
  \irule{\Gamma \vdash t:A}
  {\Gamma \vdash a \bullet t:A}
  {\mbox{prod}(a)}
  $$
  $$
  \irule{}
  {\vdash a.\star:\one}
  {\mbox{$\one$-i}(a)}
  \qquad
  \irule{\Gamma \vdash t:\one & \Delta \vdash u:A}
  {\Gamma, \Delta \vdash \elimone(t,u):A}
  {\mbox{$\one$-e}}
  $$
  $$
  \irule{\Gamma, x:A \vdash t:B}
  {\Gamma \vdash \lambda \abstr{x}t:A \multimap B}
  {\mbox{$\multimap$-i}}
  \qquad
  \irule{\Gamma \vdash t:A\multimap B & \Delta \vdash u:A}
  {\Gamma, \Delta \vdash t~u:B}
  {\mbox{$\multimap$-e}}
  $$
  $$
  \irule{\Gamma \vdash t:A & \Delta \vdash u:B}
  {\Gamma, \Delta \vdash t \otimes u:A \otimes B}
  {\mbox{$\otimes$-i}}
  \qquad
  \irule{\Gamma \vdash t:A \otimes B & \Delta, x:A, y:B \vdash u:C}
  {\Gamma, \Delta \vdash \elimtens(t,\abstr{x y}u):C}
  {\mbox{$\otimes$-e}}
  $$
  $$
  \irule{}{\Gamma \vdash \topintro:\top}{\mbox{$\top$-i}}
  \qquad
  \irule{\Gamma \vdash t:\zero}
  {\Gamma, \Delta \vdash \elimzero(t):C}
  {\mbox{$\zero$-e}}
  $$
  $$\irule{\Gamma \vdash t:A & \Gamma \vdash u:B}
  {\Gamma \vdash \pair{t}{u}:A \with B}
  {\mbox{$\with$-i}}$$
  $$\irule{\Gamma \vdash t:A \with B & \Delta, x:A \vdash u:C}
  {\Gamma, \Delta \vdash \elimwith^1(t,\abstr{x}u):C}
  {\mbox{$\with$-e1}}
  \qquad
  \irule{\Gamma \vdash t:A \with B & \Delta, x:B \vdash u:C}
  {\Gamma, \Delta \vdash \elimwith^2(t,\abstr{x}u):C}
  {\mbox{$\with$-e2}}$$
  $$\irule{\Gamma \vdash t:A}
  {\Gamma \vdash \inl(t):A \oplus B}
  {\mbox{$\oplus$-i1}}
  \qquad
  \irule{\Gamma \vdash t:B}
  {\Gamma \vdash \inr(t):A \oplus B}
  {\mbox{$\oplus$-i2}}$$
  $$\irule{\Gamma \vdash t:A \oplus B & \Delta, x:A \vdash u:C & \Delta, y:B \vdash v:C}
  {\Gamma, \Delta \vdash \elimplus(t,\abstr{x}u,\abstr{y}v):C}
  {\mbox{$\oplus$-e}}$$
  \caption{The deduction rules of the ${\mathcal L}^{\mathcal S}$-calculus.\label{figuretypingrules}}
\end{figure}

The propositions of the 
of intuitionistic multiplicative additive linear logic
are 
$$A = \one \mid A \multimap A \mid A \otimes A \mid \top \mid \zero
\mid A \with A \mid A \oplus A$$

Let ${\mathcal S}$ be a semi-ring of {\it scalars}, for instance
$\{*\}$, $\{0, 1\}$, 
${\mathbb N}$, 
${\mathbb Q}$, ${\mathbb R}$, or ${\mathbb C}$.
The proof-terms of the ${\mathcal L}^{\mathcal S}$-calculus
are 
\begin{align*}
  t =~& x \mid t \plus u \mid a \bullet t\\
  & \mid a.\star \mid \elimone(t,u)  
  \mid \lambda \abstr{x}t\mid t~u
  \mid t \otimes u \mid \elimtens(t, \abstr{x y} u)\\
  &
  \mid \topintro
  \mid \elimzero(t)
  \mid \pair{t}{u} \mid \elimwith^1(t,\abstr{x}u) \mid
  \elimwith^2(t,\abstr{x}u)
  \mid \inl(t)\mid \inr(t)\mid \elimplus(t,\abstr{x}u,\abstr{y}v)
\end{align*}
where $a$ is a scalar.

These symbols are in one-to-one correspondence with the rules of
intuitionistic multiplicative additive linear logic extended with
interstitial rules and scalars, and proof-terms can be seen as mere
one-dimensional representation of proof-trees.  The proofs of the form
$a.\star$, $\lambda \abstr{x}t$, $t \otimes u$, $\topintro$,
$\pair{t}{u}$, $\inl(t)$, and $\inr(t)$ are called {\it
  introductions}, and those of the form $\elimone(t,u)$, $t~u$,
$\elimtens(t,\abstr{xy}u)$, $\elimzero(t)$,
$\elimwith^1(t,\abstr{x}u)$, $\elimwith^2(t,\abstr{x}u)$, and
$\elimplus(t,\abstr{x}u,\abstr{y}v)$ {\it eliminations}.  The
variables and the proofs of the form $t \plus u$ and $a \bullet t$ are
neither introductions nor eliminations.

On the other hand, each symbol can be considered as a construction of a
functional programming language: the introductions are standard, while the eliminations are as follows:
\begin{itemize}
  \item $\elimone(t,u)$ is the sequence, sometimes written $t;u$,
  \item $t~u$ is the application, sometimes written $t(u)$,
  \item $\elimtens(t,\abstr{xy}u)$ is the let on pairs, sometimes written $\mathsf{let\ }(x,y)=t\ \mathsf{in}\ u$,
  \item $\elimzero(t)$ is the error, sometimes written $\mathsf{error}(t)$,
  \item $\elimwith^1(t,\abstr{x}u)$ is the first projection, sometimes written $\mathsf{let\ }x=\mathsf{fst}(t)\ \mathsf{in}\ u$,
  \item $\elimwith^2(t,\abstr{x}u)$ is the second projection, sometimes written $\mathsf{let\ }x=\mathsf{snd}(t)\ \mathsf{in}\ u$,
  \item $\elimplus(t,\abstr{x}u,\abstr{y}v)$ is the match, sometimes written $\mathsf{match}\ t\ \mathsf{in}\ \{\mathsf{inl}(x)\mapsto u\mid \mathsf{inr}(y)\mapsto v\}$.
\end{itemize}

The $\alpha$-equivalence relation and the free and bound variables of
a proof-term are defined as usual. Proof-terms are defined modulo
$\alpha$-equivalence.  A proof-term is closed if it contains no free
variables.  We write $(u/x)t$ for the substitution of $u$ for $x$ in
$t$ and if $FV(t) \subseteq \{x\}$, we also use the notation $t\{u\}$.

The typing rules are those of Figure~\ref{figuretypingrules}.  These
typing rules are exactly the deduction rules of intuitionistic linear
natural deduction, with proof-terms, with two differences: the
interstitial rules and the scalars.

\begin{figure}[t]
  \[
    \begin{array}{r@{\,}l}
      \elimone(a.\star,t) & \longrightarrow  a \bullet t\\
      (\lambda \abstr{x}t)~u & \longrightarrow  (u/x)t\\
      \elimtens(u \otimes v,\abstr{x y}w) & \longrightarrow  (u/x,v/y)w\\
      \elimwith^1(\pair{t}{u}, \abstr{x}v) & \longrightarrow  (t/x)v\\
      \elimwith^2(\pair{t}{u}, \abstr{x}v) & \longrightarrow  (u/x)v\\
      \elimplus(\inl(t),\abstr{x}v,\abstr{y}w) & \longrightarrow  (t/x)v\\
      \elimplus(\inr(u),\abstr{x}v,\abstr{y}w) & \longrightarrow  (u/y)w\\
      \\
      {a.\star} \plus b.\star&\longrightarrow  (a+b).\star\\
      (\lambda \abstr{x}t) \plus (\lambda \abstr{x}u) & \longrightarrow  \lambda \abstr{x}(t \plus u)\\
      \elimtens(t \plus u,\abstr{x y}v) & \longrightarrow 
      \elimtens(t,\abstr{x y}v)
      \plus
      \elimtens(u,\abstr{x y}v)\\
      \topintro \plus \topintro
      & \longrightarrow  \topintro
      \\
      \pair{t}{u} \plus \pair{v}{w}
      & \longrightarrow  \pair{t \plus v}{u \plus w} 
      \\
      \elimplus(t \plus u,\abstr{x}v,\abstr{y}w) & \longrightarrow 
      \elimplus(t,\abstr{x}v,\abstr{y}w)
      \plus
      \elimplus(u,\abstr{x}v,\abstr{y}w)\\
      \\
      a \bullet b.\star&\longrightarrow  (a \times b).\star\\
      a \bullet \lambda \abstr{x} t &\longrightarrow  \lambda \abstr{x} a \bullet t\\
      \elimtens(a \bullet t,\abstr{x y}v) & \longrightarrow 
      a \bullet \elimtens(t,\abstr{x y}v)\\
      a \bullet \topintro &\longrightarrow  \topintro\\
      a \bullet \pair{t}{u} &\longrightarrow  \pair{a \bullet t}{a \bullet u}\\
      \elimplus(a \bullet t,\abstr{x}v,\abstr{y}w) & \longrightarrow 
      a \bullet \elimplus(t,\abstr{x}v,\abstr{y}w)
    \end{array}
  \]
  \caption{The reduction rules of the
  ${\mathcal L}^{\mathcal S}$-calculus.  \label{figureductionrules}}
\end{figure}

The reduction rules are those of Figure~\ref{figureductionrules}.  As
usual, the reduction can occur in any context. The one-step 
reduction relation is written $\lra$, its inverse $\lla$,
its reflexive-transitive closure $\lras$, the reflexive-transitive
closure of its inverse $\llas$, and its reflexive-symmetric-transitive
closure $\equiv$.  The first seven rules correspond to the reduction of
cuts on the connectives $\one$, $\multimap$,
$\otimes$,
$\with$, and $\oplus$.
The twelve others enable to commute the interstitial rules sum and prod
with the introduction rules of the connectives $\one$, $\multimap$,
$\top$, and $\with$, and with the elimination rule of the connectives
$\otimes$ and $\oplus$.
For instance, the rule
$$\pair{t}{u} \plus \pair{v}{w} \longrightarrow  \pair{t \plus v}{u \plus w}$$
pushes the symbol $\plus$ inside the pair. The scalars are added in the rule
$${a.\star} \plus b.\star \longrightarrow  (a+b).\star$$
and multiplied in the rule
$$a \bullet b.\star \longrightarrow  (a \times b).\star$$

We now prove the subject reduction, confluence, termination, and
introduction properties of the ${\mathcal L}^{\mathcal S}$-calculus.

\subsection{Subject reduction}

The subject reduction property is not completely trivial: as noted
above, it would, for example, fail if we commuted the sum rule with
the introduction rule of the multiplicative conjunction $\otimes$.

\begin{lemma}[Substitution]
  \label{lem:substitution}
  If $\Gamma,x:B\vdash t:A$ and $\Delta\vdash u:B$, then $\Gamma,\Delta\vdash
  (u/x)t:A$.
\end{lemma}

\begin{proof}
  By induction on the structure of $t$. Since the deduction system is syntax
  directed, the generation lemma is trivial and will be implicitly used in the
  proof.
  \begin{itemize}
    \item If $t=x$, then $\Gamma=\varnothing$ and $A=B$.  Thus,
      $\Gamma,\Delta\vdash (u/x)t:A$ is the same as $\Delta\vdash
      u:B$, which is valid by hypothesis.

    \item The proof $t$ cannot be a variable $y$ different from $x$, 
      as such a variable $y$ is not a proof in $\Gamma, x:B$.

    \item If $t=v_1\plus v_2$, then $\Gamma,x:B\vdash v_1:A$ and
      $\Gamma,x:B\vdash v_2:A$. By the induction hypothesis,
      $\Gamma,\Delta\vdash (u/x)v_1:A$ and $\Gamma,\Delta \vdash
      (u/x)v_2:A$. Therefore, by the rule sum, $\Gamma,\Delta\vdash
      (u/x)v_1\plus(u/x)v_2:A$.  Hence, $\Gamma, \Delta \vdash
      (u/x)t:A$.

    \item If $t = a \bullet v$, then $\Gamma, x:B \vdash v:A$.  By
      the induction hypothesis, $\Gamma,\Delta\vdash (u/x)v:A$.
      Therefore, by the rule prod, $\Gamma,\Delta\vdash a \bullet
      (u/x)v:A$.  Hence, $\Gamma, \Delta \vdash (u/x)t:A$.

    \item The proof $t$ cannot be of the form
      $t=a.\star$, that is not a proof in $\Gamma, x:B$.

    \item If $t=\elimone(v_1,v_2)$, then $\Gamma=\Gamma_1,\Gamma_2$
      and there are two cases.
      \begin{itemize}
	\item If $\Gamma_1,x:B\vdash v_1:\one$ and $\Gamma_2\vdash
	  v_2:A$, then, by the induction hypothesis,
	  $\Gamma_1,\Delta\vdash (u/x)v_1:\one$ and, by the rule
	  $\one$-e, $\Gamma,\Delta\vdash \elimone((u/x)v_1,v_2):A$.
	\item If $\Gamma_1\vdash v_1:\one$ and
	  $\Gamma_2,x:B\vdash v_2:A$, then, by the induction hypothesis,
	  $\Gamma_2,\Delta\vdash (u/x)v_2:A$ and, by the rule $\one$-e,
	  $\Gamma,\Delta\vdash \elimone(v_1,(u/x)v_2):A$. 
      \end{itemize}
      Hence, $\Gamma, \Delta \vdash (u/x)t:A$.

    \item If $t=\lambda\abstr{y}v$, then $A=C\multimap D$ and
      $\Gamma,y:C,x:B\vdash v:D$. By the induction hypothesis,
      $\Gamma,\Delta,y:C\vdash (u/x)v:D$, so, by the rule $\multimap$-i,
      $\Gamma,\Delta\vdash\lambda\abstr{y}(u/x)v:A$.
      Hence, $\Gamma, \Delta \vdash (u/x)t:A$.
    \item If $t=v_1~v_2$, then $\Gamma=\Gamma_1,\Gamma_2$ and there are 
      two cases.
      \begin{itemize}
	\item If $\Gamma_1,x:B\vdash v_1:C\multimap A$
	  and $\Gamma_2\vdash v_2:C$, then, by the induction hypothesis,
	  $\Gamma_1,\Delta\vdash (u/x)v_1:C\multimap A$ and, by the rule
	  $\multimap$-e, $\Gamma,\Delta\vdash (u/x)v_1~v_2:A$. 
	\item If $\Gamma_1\vdash v_1:C\multimap A$ and
	  $\Gamma_2,x:B\vdash v_2:C$, then, by the induction hypothesis,
	  $\Gamma_2,\Delta\vdash (u/x)v_2:C$ and, by the rule $\multimap$-e,
	  $\Gamma,\Delta\vdash v_1~(u/x)v_2:A$. 
      \end{itemize}
      Hence, $\Gamma, \Delta \vdash (u/x)t:A$.

    \item 
      If $t=v_1\otimes v_2$, then $A=A_1\otimes A_2$, $\Gamma=\Gamma_1,\Gamma_2$, and there are two cases.
      \begin{itemize}
	\item If $\Gamma_1,x:B\vdash v_1:A_1$ and $\Gamma_2\vdash v_2:A_2$, then, by the induction hypothesis, $\Gamma_1,\Delta\vdash (u/x)v_1:A_1$, and, by the rule $\otimes$-i, $\Gamma,\Delta\vdash (u/x)v_1\otimes v_2:A$.
	\item If $\Gamma_1\vdash v_1:A_1$ and $\Gamma_2,x:B\vdash v_2:A_2$, this case is analogous to the previous one.
      \end{itemize}
      Hence, $\Gamma,\Delta\vdash (u/x)t:A$.
    \item 
      If $t=\elimtens(v_1.\abstr{yz}v_2)$, then
      $\Gamma=\Gamma_1,\Gamma_2$ and there are two cases.
      \begin{itemize}
	\item If $\Gamma_1,x:B\vdash v_1:C\otimes D$ and
	  $\Gamma_2,y:C,z:D\vdash v_2:A$. By the induction hypothesis,
	  $\Gamma_1,\Delta\vdash (u/x)v_1:C\otimes D$ and, by the rule $\otimes$-e,
	  $\Gamma,\Delta\vdash \elimtens((u/x)v_1,\abstr{yz}v_2):A$.

	\item If $\Gamma_1 \vdash v_1:C\otimes D$ and
	  $\Gamma_2,y:C,z:D,x:B\vdash v_2:A$. By the induction
	  hypothesis, $\Gamma_2,\Delta,y:C,z:D\vdash (u/x)v_2:A$ and, by the rule
	  $\otimes$-e,
	  $\Gamma,\Delta\vdash\elimtens(v_1,\abstr{yz}(u/x)v_2):A$.
      \end{itemize}
      Hence, $\Gamma, \Delta \vdash (u/x)t:A$.
    \item 
      If $t=\topintro$, then $A=\top$, and since $(u/x)t = t = \topintro$, by rule $\top$-i, $\Gamma,\Delta\vdash t:A$.

    \item If $t=\elimzero(v)$, then $\Gamma=\Gamma_1,\Gamma_2$ and
      there are two cases.
      \begin{itemize}
	\item If $\Gamma_1\vdash v:\zero$, then $x\notin FV(v)$, so
	  $(u/x)v = v$, $\Gamma_1\vdash (u/x)v:\zero$, and, by the
	  rule $\zero$-e, $\Gamma,\Delta\vdash
	  \elimzero((u/x)v):A$.

	\item If $\Gamma_1,x:B\vdash v:\zero$, then, by the
	  induction hypothesis, $\Gamma_1,\Delta \vdash
	  (u/x)v:\zero$, and, by the rule $\zero$-e,
	  $\Gamma,\Delta\vdash \elimzero((u/x)v):A$.
      \end{itemize}
      Hence, $\Gamma, \Delta \vdash (u/x)t:A$.

    \item If $t=\pair{v_1}{v_2}$, then $A=A_1\with A_2$ and
      $\Gamma,x:B\vdash v_1:A_1$ and $\Gamma,x:B\vdash v_2:A_2$. By the
      induction hypothesis, $\Gamma,\Delta\vdash (u/x)v_1:A_1$ and
      $\Gamma,\Delta \vdash (u/x)v_2:A_2$. Therefore, by the rule $\with$-i,
      $\Gamma,\Delta\vdash \pair{(u/x)v_1}{(u/x)v_2}:A$.  Hence, $\Gamma, \Delta
      \vdash (u/x)t:A$.

    \item If $t=\elimwith^1(v_1,\abstr{y}v_2)$, then
      $\Gamma=\Gamma_1,\Gamma_2$ and there are two cases.
      \begin{itemize}
	\item If $\Gamma_1,x:B\vdash v_1:C\with D$ and
	  $\Gamma_2,y:C \vdash v_2:A$. By the induction hypothesis,
	  $\Gamma_1,\Delta\vdash (u/x)v_1:C\with D$ and, by the rule $\with$-e,
	  $\Gamma,\Delta\vdash \elimwith^1((u/x)v_1,\abstr{y}v_2):A$.

	\item If $\Gamma_1 \vdash v_1:C\with D$ and
	  $\Gamma_2,y:C, x:B\vdash v_2:A$. By the induction
	  hypothesis, $\Gamma_2,\Delta,y:C \vdash (u/x)v_2:A$ and, by the rule
	  $\with$-e,
	  $\Gamma,\Delta\vdash\elimwith^1(v_1,\abstr{y}(u/x)v_2):A$.
      \end{itemize}
      Hence, $\Gamma, \Delta \vdash (u/x)t:A$.

    \item If $t=\elimwith^2(v_1,\abstr{y}v_2)$. The proof is analogous.

    \item If $t=\inl(v)$, then $A=C\oplus D$ and $\Gamma,x:B\vdash v:C$. By
      the induction hypothesis, $\Gamma,\Delta\vdash (u/x)v:C$ and so, by the rule
      $\oplus$-i, $\Gamma,\Delta\vdash \inl((u/x)v):A$.
      Hence, $\Gamma, \Delta \vdash (u/x)t:A$.

    \item If $t=\inr(v)$. The proof is analogous.

    \item If $t=\elimplus(v_1,\abstr{y}v_2,\abstr{z}v_3)$, then
      $\Gamma=\Gamma_1,\Gamma_2$ and there are two cases.
      \begin{itemize} \item If 
	  $\Gamma_1,x:B \vdash v_1:C\oplus D$, $\Gamma_2,y:C \vdash
	  v_2:A$, and $\Gamma_2,z:D \vdash v_3:A$. By the
	  induction hypothesis, $\Gamma_1,\Delta\vdash (u/x)v_1:C\oplus D$
	  and, by the rule $\oplus$-e, $\Gamma,\Delta\vdash
	  \elimplus((u/x)v_1,\abstr{y}v_2,\abstr{z}v_3):A$.  

	\item If $\Gamma_1\vdash v_1:C\oplus D$, $\Gamma_2,y:C,x:B\vdash
	  v_2:A$, and $\Gamma_2,z:D,x:B\vdash v_3:A$.  By the
	  induction hypothesis, $\Gamma_2,\Delta,y:C,\vdash
	  (u/x)v_2:A$ and $\Gamma_2,\Delta,z:D\vdash
	  (u/x)v_3:A$. By the rule $\oplus$-e,
	  $\Gamma,\Delta\vdash\elimplus(v_1,\abstr{y}(u/x)v_2,\abstr{z}(u/x)v_3):A$.
      \end{itemize}
      Hence, $\Gamma, \Delta \vdash (u/x)t:A$.
      \qedhere
  \end{itemize}
\end{proof}

\begin{theorem}[Subject reduction]
  \label{th:subjectreduction}
  If $\Gamma \vdash t:A$ and $t \lra u$, then $\Gamma \vdash u:A$.
\end{theorem}

\begin{proof}
  By induction on the definition of the relation $\lra$. The context
  cases are trivial, so we focus on the reductions at top level. As the
  generation lemma is trivial, we use it implicitly in the proof.

  \begin{itemize}

    \item If $t = \elimone(a.\star,v)$ and $u = a \bullet v$, then
      $\vdash a.\star:\one$ and $\Gamma \vdash v:A$. Hence,
      $\Gamma \vdash a \bullet v:A$.

    \item If $t = (\lambda \abstr{x} v_1)~v_2$ and $u = (v_2/x)v_1$, then
      $\Gamma = \Gamma_1, \Gamma_2$, $\Gamma_1, x:B \vdash v_1:A$, and
      $\Gamma_2 \vdash v_2:B$. By Lemma~\ref{lem:substitution}, $\Gamma
      \vdash u:A$.

    \item 
      If $t=\elimtens(v_1\otimes v_2,\abstr{xy}v_3)$ and $u=(v_1/x,v_2/y)v_3$,
      then $\Gamma=\Gamma_1,\Gamma_2,\Gamma_3$, $\Gamma_1\vdash v_1:B_1$,
      $\Gamma_2\vdash v_2:B_2$, and $\Gamma_3,x:B_1,y:B_2\vdash v_3:A$.  By
      Lemma~\ref{lem:substitution}, $\Gamma_2,\Gamma_3,x:B_1\vdash (v_2/y)v_3:A$,
      and, by Lemma~\ref{lem:substitution} again, $\Gamma\vdash
      (v_1/x,v_2/y)v_3:A$. That is, $\Gamma\vdash u:A$.
    \item If $t = \elimwith^1(\pair{v_1}{v_2}, \abstr{y}v_3)$ and $u = 
      (v_1/y)v_3$, then $\Gamma = \Gamma_1, \Gamma_2$, $\Gamma_1 \vdash
      v_1:B$, $\Gamma_1 \vdash v_2:C$, and $\Gamma_2, y:B \vdash v_3:A$.
      By Lemma~\ref{lem:substitution}, $\Gamma \vdash (v_1/y)v_3:A$, that
      is $\Gamma \vdash u:A$.

    \item If $t = \elimwith^2(\pair{v_1}{v_2}), \abstr{y}v_3)$ and $u =
    (v_2/y)v_3$, the proof is analogous.

  \item If $t = \elimplus(\inl(v_1), \abstr{y} v_2, \abstr{z} v_3)$ and $u
    = (v_1/y)v_2$ then $\Gamma = \Gamma_1, \Gamma_2$, $\Gamma_1 \vdash
    v_1:B$, and $\Gamma_2, y:B \vdash v_2:A$.  By Lemma~\ref{lem:substitution},
    $\Gamma \vdash u:A$.

  \item If $t = \elimplus(\inr(v_1), \abstr{y} v_2, \abstr{z} v_3)$ and $u
    = (v_1/z)v_3$, the proof is analogous.

  \item If $t = {a.\star} \plus b.\star$ and $u = (a+b).\star$, then $A
    = \one$ and $\Gamma$ is empty. Thus, $\Gamma \vdash u:A$.

  \item If $t = \lambda \abstr{x} v_1 \plus \lambda \abstr{x} v_2$ and
    $u = \lambda \abstr{x} (v_1 \plus v_2)$ then $A = B \multimap C$,
    $\Gamma, x:B \vdash v_1:C$, and $\Gamma, x:B \vdash v_2:C$.  Thus,
    $\Gamma \vdash u:A$.

  \item 
    If $t=\elimtens(v_1\plus v_2,\abstr{xy}v_3)$ and
    $u=\elimtens(v_1,\abstr{xy}v_3)\plus\elimtens(v_2,\abstr{xy}v_3)$, then
    $\Gamma = \Gamma_1, \Gamma_2$, $\Gamma_1 \vdash v_1:B \otimes C$,
    $\Gamma_1 \vdash v_2:B \otimes C$, and $\Gamma_2, x:B,y:C \vdash v_3:A$.
    Hence, $\Gamma \vdash u:A$.
  \item 
    If $t=\topintro\plus\topintro$ and $u=\topintro$, then
    $\Gamma\vdash\topintro :A$, that is, $\Gamma\vdash u:A$.
  \item If $t = \pair{v_1}{v_2} \plus \pair{v_3}{v_4}$ and $u =
    \pair{v_1 \plus v_3}{v_2 \plus v_4}$ then $A = B \with C$, $\Gamma
    \vdash v_1:B$, $\Gamma \vdash v_2:C$, and $\Gamma \vdash v_3:B$,
    $\Gamma \vdash v_4:C$. Thus, $\Gamma \vdash u:A$.

  \item If $t = \elimplus(v_1 \plus v_2,\abstr{x}v_3,\abstr{y}v_4)$ and $u
    = \elimplus(v_1,\abstr{x}v_3,\abstr{y}v_4) \plus
    \elimplus(v_2,\abstr{x}v_3,\abstr{y}v_4)$ then $\Gamma = \Gamma_1,
    \Gamma_2$, $\Gamma_1 \vdash v_1:B \oplus C$, $\Gamma_1 \vdash v_2:B
    \oplus C$, $\Gamma_2, x:B \vdash v_3:A$, $\Gamma_2, y:C \vdash v_4:A$.
    Hence, $\Gamma \vdash u:A$.

  \item If $t = a \bullet b.\star$ and $u = (a \times b).\star$, then $A
    = \one$, $\Gamma$ is empty. Thus, $\Gamma \vdash u:A$.

  \item If $t = a \bullet \lambda \abstr{x} v$ and $u = \lambda
    \abstr{x} a \bullet v$, then $A = B \multimap C$ and $\Gamma, x:B
    \vdash v:C$. Thus, $\Gamma \vdash u:A$.

  \item 
    If $t=\elimtens(a\bullet v_1,\abstr{xy}v_2)$ and
    $u=a\bullet\elimtens(v_1,\abstr{xy}v_2)$, then $\Gamma = \Gamma_1,
    \Gamma_2$, $\Gamma_1 \vdash v_1:B \otimes C$ and, $\Gamma_2, x:B,y:C
    \vdash v_2:A$. Thus, $\Gamma \vdash u:A$. 
  \item 
    If $t=a\bullet\topintro$ and $u=\topintro$, then
    $\Gamma\vdash\topintro:A$. Thus, $\Gamma\vdash u:A$.
  \item If $t = a \bullet \pair{v_1}{v_2}$ and $u = \pair{a \bullet
    v_1}{a \bullet v_2}$, then $A = B \with C$, $\Gamma \vdash v_1:B$,
    and $\Gamma \vdash v_2:C$. Thus, $\Gamma \vdash u:A$.

  \item If $t = \elimplus(a \bullet v_1,\abstr{x}v_2,\abstr{y}v_3)$ and $u
    = a \bullet \elimplus(v_1,\abstr{x}v_2,\abstr{y}v_3)$, then $\Gamma =
    \Gamma_1, \Gamma_2$, $\Gamma_1 \vdash v_1:B \oplus C$, $\Gamma_2, x:B
    \vdash v_2:A$, and $\Gamma_2, y:C \vdash v_3:A$. Thus, $\Gamma \vdash
    u:A$.  \qedhere
\end{itemize}
\end{proof}

\subsection{Confluence}

\begin{theorem}[Confluence]
  \label{th:Confluence}
The ${\mathcal L}^{\mathcal S}$-calculus is confluent, that is
whenever $u \llas t \lras v$, then there exists a $w$,
such as $u \lras w \llas v$.
\end{theorem}

\begin{proof}
  The reduction rules of Figure~\ref{figureductionrules} applied to
  well-formed proofs is left linear and has no critical pairs~\cite[Section 6]{Nipkow}.
  By
  \cite[Theorem 6.8]{Nipkow}, it is confluent.
\end{proof}

\subsection{Termination}

We now prove that the ${\mathcal L}^{\mathcal S}$-calculus strongly
terminates, that is that all reduction sequences are finite.  To
handle the symbols $\plus$ and $\bullet$ and the associated reduction
rules, we prove the strong termination of an extended reduction
system, in the spirit of Girard's ultra-reduction \cite{Girard},
whose strong termination obviously implies that of the rules of
Figure~\ref{figureductionrules}.

\begin{definition}[Ultra-reduction]
  Ultra-reduction is defined with the rules of Figure~\ref{figureductionrules},
  plus the rules
  \begin{align*}
    t \plus u & \longrightarrow  t\\
    t \plus u & \longrightarrow  u\\
    a \bullet t & \longrightarrow  t
  \end{align*}
\end{definition}

\begin{definition}[Length of reduction]
  If $t$ is a strongly terminating proof, we write $\length{t}$ for the
  maximum length of a reduction sequence issued from $t$.
\end{definition}

\begin{lemma}[Termination of a sum]
  \label{terminationsum}
  If $t$ and $u$ strongly terminate, then so does $t \plus u$. 
\end{lemma}

\begin{proof}
  We prove that all the one-step reducts of 
  $t \plus u$ strongly terminate, by induction first on 
  $\length{t} + \length{u}$ and then on the size of $t$. 

  If the reduction takes place in $t$ or in $u$ we apply the induction
  hypothesis.
  Otherwise, the reduction is at the root and the rule used is either
  \begin{align*}
    {a.\star} \plus {b.\star} &\longrightarrow  (a + b).\star\\
    (\lambda \abstr{x}t') \plus (\lambda \abstr{x}u')
    &\longrightarrow  \lambda \abstr{x}(t' \plus u')\\
    \topintro \plus \topintro &\longrightarrow \topintro\\
    \pair{t'_1}{t'_2} \plus \pair{u'_1}{u'_2}
    &\longrightarrow  \pair{t'_1 \plus u'_1}{t'_2 \plus u'_2}\\
    t \plus u &\longrightarrow t\\
    t \plus u &\longrightarrow u
  \end{align*}
  In the first case, the proof $(a + b).\star$ is irreducible, hence it
  strongly terminates. In the second, 
  by induction hypothesis, the proof $t' \plus u'$
  strongly terminates, thus so does the proof
  $\lambda \abstr{x}(t' \plus u')$.
  In the third, the proof $\topintro$ is irreducible, hence it
  strongly terminates. 
  In the fourth, 
  by induction hypothesis, the proofs
  $t'_1 \plus u'_1$ and $t'_2 \plus u'_2$
  strongly terminate, hence so does the proof
  $\pair{t'_1 \plus u'_1}{t'_2 \plus u'_2}$.
  In the fifth and the sixth, the proofs $t$ and $u$ strongly terminate. 
\end{proof}

\begin{lemma}[Termination of a product]
  \label{terminationprod}
  If $t$ strongly terminates, then so does $a \bullet t$. 
\end{lemma}

\begin{proof}
  We prove that all the one-step reducts of 
  $a \bullet t$ strongly terminate, by induction first on 
  $\length{t}$ and then on the size of $t$. 

  If the reduction takes place in $t$, we apply the induction
  hypothesis.
  Otherwise, the reduction is at the root and the rule used is either
  \begin{align*}
    a \bullet b.\star &\longrightarrow  (a \times b).\star\\
    a \bullet (\lambda \abstr{x}t') 
    &\longrightarrow  \lambda \abstr{x} a \bullet t'\\
    a \bullet \topintro &\longrightarrow  \topintro\\
    a \bullet \pair{t'_1}{t'_2} 
    &\longrightarrow  \pair{a \bullet t'_1}{a \bullet t'_2}\\
    a \bullet t &\longrightarrow t
  \end{align*}
  In the first case, the proof $(a \times b).\star$ is irreducible,
  hence it strongly terminates. In the second, by induction
  hypothesis, the proof $a \bullet t'$ strongly terminates, thus so
  does the proof $\lambda \abstr{x} a \bullet t'$.  In the third, the
  proof $\topintro$ is irreducible, hence it strongly terminates.
  In the fourth, by induction hypothesis, the proofs $a \bullet t'_1$
  and $a \bullet t'_2$ strongly terminate, hence so does the proof
  $\pair{a \bullet t'_1}{a \bullet t'_2}$. In the fifth, the proof $t$
  strongly terminates.
\end{proof}

\begin{definition}
  Let $\mathsf{ST}$ be the set of strongly terminating terms.
  We define, by induction on the proposition $A$, a set of proofs
  $\llbracket A \rrbracket$:
  \begin{align*}
    \llbracket \one \rrbracket&=\mathsf{ST}\\
    \llbracket A \multimap B \rrbracket &=\{t\in\mathsf{ST}\mid \textrm{If }t\to^*\lambda \abstr{x}u\textrm{ then for all }v \in \llbracket A \rrbracket, (v/x)u \in \llbracket B \rrbracket\}\\
    \llbracket A \otimes B \rrbracket &= \{t\in\mathsf{ST}\mid \textrm{If }t\to^*u \otimes v\textrm{ then }u \in \llbracket A \rrbracket\textrm{ and }v \in \llbracket B \rrbracket\}\\
    \llbracket \top \rrbracket &=\mathsf{ST}\\
    \llbracket \zero \rrbracket &=\mathsf{ST}\\
    \llbracket A \with B \rrbracket &=\{t\in\mathsf{ST}\mid\textrm{If }t\to^* \pair{u}{v}\textrm{ then }u \in \llbracket A \rrbracket\textrm{ and }v \in \llbracket B \rrbracket\}\\
    \llbracket A \oplus B \rrbracket &=
    \{t\in\mathsf{ST}\mid\textrm{If }t\to^*\inl(u)\textrm{ then } u \in \llbracket A \rrbracket \textrm{ and if }t\to^*\inr(v)\textrm{ then } v \in \llbracket B \rrbracket\}
  \end{align*}
\end{definition}

\begin{lemma}[Variables]
  \label{Var}
  For any $A$, the set $\llbracket A \rrbracket$ contains all the variables.
\end{lemma}

\begin{proof}
  A variable is irreducible, hence it strongly terminates. Moreover, it
  never reduces to an introduction.
\end{proof}   

\begin{lemma}[Closure by reduction]
  \label{closure}
  If $t \in \llbracket A \rrbracket$ and $t \longrightarrow^* t'$, then 
  $t' \in \llbracket A \rrbracket$.
\end{lemma}

\begin{proof}
  If $t \longrightarrow^* t'$ and $t$ strongly terminates, then $t'$
  strongly terminates.

  Furthermore, if $A$ has the form $B \multimap C$ and $t'$ reduces to
  $\lambda \abstr{x}u$, then so does $t$, hence for every $v \in \llbracket B
  \rrbracket$, $(v/x)u \in \llbracket C \rrbracket$.
  If $A$ has the form $B \otimes C$ and $t'$ reduces to $u \otimes v$,
  then so does $t$, hence $u \in \llbracket B \rrbracket$ and $v \in
  \llbracket C \rrbracket$.
  If $A$ has the form $B \with C$ and $t'$ reduces to $\pair{u}{v}$,
  then so does $t$, hence $u \in \llbracket B \rrbracket$ and $v \in
  \llbracket C \rrbracket$.
  If $A$ has the form $B \oplus C$ and $t'$ reduces to $\inl(u)$, then so
  does $t$, hence $u \in \llbracket B \rrbracket$.
  And if $A$ has the
  form $B \oplus C$ and $t'$ reduces to $\inr(v)$, then so does $t$, hence
  $v \in \llbracket C \rrbracket$.
\end{proof}

\begin{lemma}[Girard's lemma]
  \label{CR3}
  Let $t$ be a proof that is not an introduction, 
  such that all the one-step reducts of $t$
  are in $\llbracket A \rrbracket$. Then $t \in \llbracket A \rrbracket$.
\end{lemma}

\begin{proof}
  Let $t, t_2, \dots$ be a reduction sequence issued from $t$. If it has a
  single element, it is finite. Otherwise, we have $t \longrightarrow
  t_2$. As $t_2 \in \llbracket A \rrbracket$, it strongly terminates and
  the reduction sequence is finite. Thus, $t$ strongly terminates.

  Furthermore, if $A$ has the form $B \multimap C$ and $t
  \longrightarrow^* \lambda \abstr{x}u$, then let $t , t_2, ..., t_n$ be a
  reduction sequence from $t$ to $\lambda \abstr{x}u$.  As $t_n$ is an
  introduction and $t$ is not, $n \geq 2$. Thus, $t \longrightarrow t_2
  \longrightarrow^* t_n$. We have $t_2 \in \llbracket A \rrbracket$,
  thus for all $v \in \llbracket B \rrbracket$, $(v/x)u \in \llbracket C
  \rrbracket$.

  If $A$ has the form $B \otimes C$ and $t \longrightarrow^* u \otimes v$, then let $t , t_2, ..., t_n$ be a reduction sequence
  from $t$ to $u \otimes v$.  As $t_n$ is an introduction and
  $t$ is not, $n \geq 2$. Thus, $t \longrightarrow t_2 \longrightarrow^*
  t_n$. We have $t_2 \in \llbracket A \rrbracket$, thus $u \in
  \llbracket B \rrbracket$ and $v \in \llbracket C \rrbracket$.

  If $A$ has the form $B \with C$ and $t \longrightarrow^* \pair{u}{v}$,
  the proof is similar.

  If $A$ has the form $B \oplus C$ and $t \longrightarrow^* \inl(u)$, then let $t,
  t_2,\dots, t_n$ be a reduction sequence from $t$ to $\inl(u)$.  As
  $t_n$ is an introduction and $t$ is not, $n \geq 2$.  Thus, $t
  \longrightarrow t_2 \longrightarrow^* t_n$. We have $t_2 \in
  \llbracket A \rrbracket$, thus $u \in \llbracket B \rrbracket$.

If $A$ has the form $B \oplus C$ and $t \longrightarrow^*
 \inr(v)$, the proof is similar.
 \qedhere
\end{proof}

In Lemmas~\ref{sum} to~\ref{elimplus}, we prove the adequacy of each proof constructor. 

\begin{lemma}[Adequacy of $\plus$]
  \label{sum}
  If $t_1 \in \llbracket A \rrbracket$ and $t_2 \in \llbracket A
  \rrbracket$, then $t_1 \plus t_2 \in \llbracket A \rrbracket$.
\end{lemma}

\begin{proof}
  By induction on $A$.
  The proofs $t_1$ and $t_2$ strongly terminate.  Thus, by
  Lemma~\ref{terminationsum}, the proof $t_1 \plus t_2$ strongly
  terminates.
  Furthermore:

  \begin{itemize}
    \item If the proposition $A$ has the form $B \multimap C$, and $t_1
      \plus t_2 \lras \lambda \abstr{x} v$ then either $t_1 \lras
      \lambda \abstr{x} u_1$, $t_2 \lras \lambda \abstr{x} u_2$, and $u_1
      \plus u_2 \lras v$, or $t_1 \lras \lambda \abstr{x} v$, or $t_2
      \lras \lambda \abstr{x} v$.

      In the first case, as $t_1$ and $t_2$ are in $\llbracket A
      \rrbracket$, for every $w$ in $\llbracket B \rrbracket$, $(w/x)u_1
      \in \llbracket C \rrbracket$ and $(w/x)u_2 \in \llbracket C
      \rrbracket$.  By induction hypothesis, $(w/x)(u_1 \plus u_2) =
      (w/x)u_1 \plus (w/x)u_2 \in \llbracket C \rrbracket$ and by
      Lemma~\ref{closure}, $(w/x)v \in \llbracket C \rrbracket$.

      In the second and the third, as $t_1$ and $t_2$ are in $\llbracket A
      \rrbracket$, for every $w$ in $\llbracket B \rrbracket$, $(w/x)v \in
      \llbracket C \rrbracket$.

    \item If the proposition $A$ has the form $B \otimes C$, and $t_1
      \plus t_2 \lras v \otimes v'$ then $t_1 \lras v \otimes v'$, or
      $t_2 \lras v \otimes v'$.  As $t_1$ and $t_2$ are in $\llbracket
      A \rrbracket$, $v \in \llbracket B \rrbracket$ and $v' \in
      \llbracket C \rrbracket$.

    \item If the proposition $A$ has the form $B \with C$, and $t_1 \plus t_2
      \lras \pair{v}{v'}$ then $t_1 \lras \pair{u_1}{u'_1}$, $t_2 \lras
      \pair{u_2}{u'_2}$, $u_1 \plus u_2 \lras v$, and $u'_1 \plus u'_2
      \lras v'$, or $t_1 \lras \pair{v}{v'}$, or $t_2 \lras
      \pair{v}{v'}$.

      In the first case, as $t_1$ and $t_2$ are in $\llbracket A
      \rrbracket$, $u_1$ and $u_2$ are in $\llbracket B \rrbracket$ and
      $u'_1$ and $u'_2$ are in $\llbracket C \rrbracket$.  By induction
      hypothesis, $u_1 \plus u_2 \in \llbracket B \rrbracket$ and $u'_1
      \plus u'_2 \in \llbracket C \rrbracket$ and by
      Lemma~\ref{closure}, $v \in \llbracket B \rrbracket$ and $v'
      \in \llbracket C \rrbracket$.

      In the second and the third, as $t_1$ and $t_2$ are in $\llbracket A
      \rrbracket$, $v \in \llbracket B \rrbracket$ and $v' \in \llbracket
      C \rrbracket$.

    \item If the proposition $A$ has the form $B \oplus C$, and $t_1 \plus
      t_2 \lras \inl(v)$ then $t_1 \lras \inl(v)$ or $t_2 \lras
      \inl(v)$.  As $t_1$ and $t_2$ are in $\llbracket A \rrbracket$, $v
      \in \llbracket B \rrbracket$.

      The proof is similar if $t_1 \plus t_2 \lras \inr(v)$.  \qedhere
  \end{itemize}
\end{proof}

\begin{lemma}[Adequacy of $\bullet$]
  \label{prod}
  If $t \in \llbracket A \rrbracket$, then $a \bullet t \in \llbracket A
  \rrbracket$.
\end{lemma}

\begin{proof}
  By induction on $A$.  The proof $t$ strongly terminates.  Thus, by
  Lemma~\ref{terminationprod}, the proof $a \bullet t$ strongly
  terminates.  Furthermore:

  \begin{itemize}
    \item If the proposition $A$ has the form $B \multimap C$, and $a
      \bullet t \lras \lambda \abstr{x} v$ then either $t \lras \lambda
      \abstr{x} u$ and $a \bullet u \lras v$, or $t \lras \lambda
      \abstr{x} v$.

      In the first case, as $t$ is in $\llbracket A \rrbracket$, for
      every $w$ in $\llbracket B \rrbracket$, $(w/x)u \in \llbracket C
      \rrbracket$.  By induction hypothesis, $(w/x) (a \bullet u) = a
      \bullet (w/x)u \in \llbracket C \rrbracket$ and by
      Lemma~\ref{closure}, $(w/x)v \in \llbracket C \rrbracket$.

      In the second, as $t$ is in $\llbracket A \rrbracket$, for every $w$
      in $\llbracket B \rrbracket$, $(w/x)v \in \llbracket C \rrbracket$.

    \item If the proposition $A$ has the form $B \otimes C$, and $a
      \bullet t \lras v \otimes v'$ then $t \lras v \otimes v'$.  As $t$
      is in $\llbracket A \rrbracket$, $v \in \llbracket B \rrbracket$ and
      $v' \in \llbracket C \rrbracket$.

    \item If the proposition $A$ has the form $B \with C$, and $a \bullet t
      \lras \pair{v}{v'}$ then $t \lras \pair{u}{u'}$, $a \bullet u
      \lras v$, and $a \bullet u' \lras v'$, or $t \lras \pair{v}{v'}$.

      In the first case, as $t$ is in $\llbracket A \rrbracket$, $u$ is in
      $\llbracket B \rrbracket$ and $u'$ is in $\llbracket C \rrbracket$.
      By induction hypothesis, $a \bullet u \in \llbracket B \rrbracket$
      and $a \bullet u' \in \llbracket C \rrbracket$ and by
      Lemma~\ref{closure}, $v \in \llbracket B \rrbracket$ and $v'
      \in \llbracket C \rrbracket$.

      In the second, as $t$ is in $\llbracket A \rrbracket$, $v \in
      \llbracket B \rrbracket$ and $v' \in \llbracket C \rrbracket$.

    \item If the proposition $A$ has the form $B \oplus C$, and $a \bullet t
      \lras \inl(v)$ then $t \lras \inl(v)$.
      Then, by Lemma~\ref{closure},
      $\inl(v) \in \llbracket A \rrbracket$ hence, $v \in \llbracket
      B \rrbracket$.

      The proof is similar if $a \bullet t \lras \inr(v)$.
      \qedhere
  \end{itemize}
\end{proof}

\begin{lemma}[Adequacy of $a.\star$]
  \label{star}
  We have $a.\star \in \llbracket \one \rrbracket$.
\end{lemma}

\begin{proof}
  As $a.\star$ is irreducible, it strongly terminates, hence
  $a.\star \in \llbracket \one \rrbracket$. 
\end{proof}

\begin{lemma}[Adequacy of $\lambda$]
  \label{abstraction}
  If, for all $u \in \llbracket A \rrbracket$, $(u/x)t \in \llbracket B
  \rrbracket$, then $\lambda \abstr{x}t \in \llbracket A \multimap B
  \rrbracket$.
\end{lemma}

\begin{proof}
  By Lemma~\ref{Var}, $x \in \llbracket A \rrbracket$, thus
  $t = (x/x)t \in \llbracket B \rrbracket$. Hence, $t$ strongly
  terminates.  Consider a reduction sequence issued from $\lambda
  \abstr{x}t$.  This sequence can only reduce $t$ hence it is finite. Thus,
  $\lambda \abstr{x}t$ strongly terminates.

  Furthermore, if $\lambda \abstr{x}t \longrightarrow^* \lambda \abstr{x}t'$, then
  $t \lras t'$.  Let $u \in \llbracket A \rrbracket$,
  $(u/x)t \lras (u/x)t'$.
  As $(u/x)t \in \llbracket B
  \rrbracket$, by Lemma~\ref{closure}, $(u/x)t' \in
  \llbracket B \rrbracket$.
\end{proof}

\begin{lemma}[Adequacy of $\otimes$]
  \label{tensor}
  If $t_1 \in \llbracket A \rrbracket$ and $t_2 \in \llbracket B
  \rrbracket$, then $t_1 \otimes t_2 \in \llbracket A \otimes B
  \rrbracket$.
\end{lemma}

\begin{proof}
  The proofs $t_1$ and $t_2$ strongly terminate. Consider a reduction
  sequence issued from $t_1 \otimes t_2$.  This sequence can only reduce
  $t_1$ and $t_2$, hence it is finite.  Thus, $t_1 \otimes t_2$ strongly
  terminates.

  Furthermore, if $t_1 \otimes t_2 \longrightarrow^* t'_1 \otimes t'_2$,
  then $t_1 \lras t'_1$ and $t_2 \lras t'_2$.  By Lemma~\ref{closure},
  $t'_1 \in \llbracket A \rrbracket$ and $t'_2 \in
  \llbracket B \rrbracket$.
\end{proof}

\begin{lemma}[Adequacy of $\topintro$]
  \label{unit}
  We have $\topintro \in \llbracket \top \rrbracket$.
\end{lemma}

\begin{proof}
  As $\topintro$ is irreducible, it strongly terminates, hence
  $\topintro \in \llbracket \top \rrbracket$. 
\end{proof}

\begin{lemma}[Adequacy of $\pair{.}{.}$]
  \label{pair}
  If $t_1 \in \llbracket A \rrbracket$ and $t_2 \in \llbracket B
  \rrbracket$, then $\pair{t_1}{t_2} \in \llbracket A \with B
  \rrbracket$.
\end{lemma}

\begin{proof}
  The proofs $t_1$ and $t_2$ strongly terminate. Consider a reduction
  sequence issued from $\pair{t_1}{t_2}$.  This sequence can only
  reduce $t_1$
  and $t_2$, hence it is finite.  Thus, $\pair{t_1}{t_2}$
  strongly terminates.

  Furthermore, if $\pair{t_1}{t_2} \longrightarrow^*
  \pair{t'_1}{t'_2}$, then $t_1 \lras t'_1$ and $t_2 \lras t'_2$.
  By Lemma~\ref{closure}, $t'_1 \in \llbracket A \rrbracket$ and $t'_2
  \in \llbracket B \rrbracket$.
\end{proof}

\begin{lemma}[Adequacy of $\inl$]
  \label{inl}
  If $t \in \llbracket A \rrbracket$, then $\inl(t) \in \llbracket A
  \oplus B \rrbracket$.
\end{lemma}

\begin{proof}
  The proof $t$ strongly terminates. Consider a reduction sequence
  issued from $\inl(t)$.  This sequence can only reduce $t$, hence it is
  finite.  Thus, $\inl(t)$ strongly terminates.

  Furthermore, if $\inl(t) \longrightarrow^* \inl(t')$, then $t \lras t'$.  By Lemma~\ref{closure}, $t' \in \llbracket A
  \rrbracket$. And $\inl(t)$ never reduces to $\inr(t')$. 
\end{proof}

\begin{lemma}[Adequacy of $\inr$]
  \label{inr}
  If $t \in \llbracket B \rrbracket$, then $\inr(t) \in \llbracket A
  \oplus B \rrbracket$.
\end{lemma}

\begin{proof}
  Similar to the proof of Lemma~\ref{inl}.  
\end{proof}

\begin{lemma}[Adequacy of $\elimone$]
  \label{elimone}
  If $t_1 \in \llbracket \one \rrbracket$ and $t_2 \in \llbracket C \rrbracket$, 
  then $\elimone(t_1,t_2) \in \llbracket C \rrbracket$.
\end{lemma}

\begin{proof}
  The proofs $t_1$ and $t_2$ strongly terminate.  We prove, by
  induction on $\length{t_1} + \length{t_2}$, that $\elimone(t_1,t_2)
  \in \llbracket C \rrbracket$.  Using Lemma~\ref{CR3}, we only
  need to prove that every of its one step reducts is in $\llbracket C
  \rrbracket$.  If the reduction takes place in $t_1$ or $t_2$, then we
  apply Lemma~\ref{closure} and the induction hypothesis.

  Otherwise, the proof $t_1$ is $a.\star$ and the
  reduct is $a \bullet t_2$. We conclude with Lemma~\ref{prod}.
\end{proof}

\begin{lemma}[Adequacy of application]
  \label{application}
  If $t_1 \in \llbracket A \multimap B \rrbracket$ and $t_2 \in
  \llbracket A \rrbracket$, then $t_1~t_2 \in \llbracket B
  \rrbracket$.
\end{lemma}

\begin{proof}
  The proofs $t_1$ and $t_2$ strongly terminate. We prove, by induction
  on $\length{t_1} + \length{t_2}$, that $t_1~t_2 \in \llbracket B \rrbracket$. Using
  Lemma~\ref{CR3}, we only need to prove that every of its one
  step reducts is in $\llbracket B \rrbracket$.  If the reduction takes
  place in $t_1$ or in $t_2$, then we apply Lemma~\ref{closure}
  and the induction hypothesis.

  Otherwise, the proof $t_1$ has the form $\lambda \abstr{x}u$ and the reduct
  is $(t_2/x)u$.  As $\lambda \abstr{x}u \in \llbracket A \multimap B
  \rrbracket$, we have $(t_2/x)u \in \llbracket B \rrbracket$.
\end{proof}

\begin{lemma}[Adequacy of $\elimtens$]
  \label{elimtens}
  If $t_1 \in \llbracket A \otimes B \rrbracket$ and for all $u$ in
  $\llbracket A \rrbracket$, for all $v$ in $\llbracket B \rrbracket$,
  $(u/x,v/y)t_2 \in \llbracket C \rrbracket$, then $\elimtens(t_1,
  \abstr{xy}t_2) \in \llbracket C \rrbracket$.
\end{lemma}

\begin{proof}
  By Lemma~\ref{Var}, $x \in \llbracket A \rrbracket$ and $y \in
  \llbracket B \rrbracket$, thus $t_2 = (x/x,y/y)t_2 \in \llbracket C
  \rrbracket$.  Hence, $t_1$ and $t_2$ strongly terminate.  We prove, by
  induction on $\length{t_1} + \length{t_2}$, that $\elimtens(t_1,
  \abstr{xy}t_2) \in \llbracket C \rrbracket$.  Using
  Lemma~\ref{CR3}, we only need to prove that every of its one
  step reducts is in $\llbracket C \rrbracket$.  If the reduction takes
  place in $t_1$ or $t_2$, then we apply Lemma~\ref{closure} and the
  induction hypothesis. Otherwise, either:
  \begin{itemize}
    \item The proof $t_1$ has the form $w_2 \otimes w_3$ and the reduct is
      $(w_2/x,w_3/y)t_2$. As
      $w_2 \otimes w_3 \in \llbracket A \otimes B \rrbracket$, we
      have $w_2 \in \llbracket A \rrbracket$
      and $w_3 \in \llbracket B \rrbracket$. 
      Hence, $(w_2/x,w_3/y)t_2 \in
      \llbracket C \rrbracket$.

    \item The proof $t_1$ has the form $t_1' \plus t''_1$ and the
      reduct is $\elimtens(t'_1, \abstr{xy}t_2) \plus
      \elimtens(t''_1, \abstr{xy}t_2)$. As $t_1
      \longrightarrow t'_1$ with an ultra-reduction rule, we have by
      Lemma~\ref{closure}, $t'_1 \in \llbracket A \otimes B
      \rrbracket$.  Similarly, $t''_1 \in \llbracket A \otimes B
      \rrbracket$.  Thus, by induction hypothesis, $\elimtens(t'_1,
      \abstr{xy}t_2) \in \llbracket A \otimes B \rrbracket$
      and $\elimtens(t''_1, \abstr{xy}t_2) \in \llbracket A
      \otimes B \rrbracket$.  We conclude with Lemma~\ref{sum}.

    \item The proof $t_1$ has the form $a \bullet t_1'$ and the
      reduct is $a \bullet \elimtens(t'_1, \abstr{xy}t_2)$. As $t_1
      \longrightarrow t'_1$ with an ultra-reduction rule, we have by
      Lemma~\ref{closure}, $t'_1 \in \llbracket A \oplus B
      \rrbracket$.  
      Thus, by induction hypothesis, $\elimtens(t'_1,
      \abstr{xy}t_2) \in \llbracket A \otimes B \rrbracket$.
      We conclude with Lemma~\ref{prod}.
      \qedhere
  \end{itemize}
\end{proof}

\begin{lemma}[Adequacy of $\elimzero$]
  \label{elimzero}
  If $t \in \llbracket \zero \rrbracket$, 
  then $\elimzero(t) \in \llbracket C \rrbracket$.
\end{lemma}

\begin{proof}
  The proof $t$ strongly terminates.  Consider a reduction sequence
  issued from $\elimzero(t)$.  This sequence can only reduce $t$, hence it
  is finite.  Thus, $\elimzero(t)$ strongly terminates.  Moreover, it
  never reduces to an introduction.
\end{proof}

\begin{lemma}[Adequacy of $\elimwith^1$]
  \label{elimwith1}
  If $t_1 \in \llbracket A \with  B \rrbracket$
  and, for all $u$ in $\llbracket A \rrbracket$,
  $(u/x)t_2 \in \llbracket C \rrbracket $, 
  then $\elimwith^1(t_1, \abstr{x}t_2) \in \llbracket C \rrbracket$.
\end{lemma}

\begin{proof}
  By Lemma~\ref{Var}, $x \in \llbracket A \rrbracket$
  thus $t_2 = (x/x)t_2 \in \llbracket C
  \rrbracket$.  Hence, $t_1$ and $t_2$ strongly terminate.  We prove, by
  induction on $\length{t_1} + \length{t_2}$, that $\elimwith^1(t_1, \abstr{x}t_2)
  \in \llbracket C \rrbracket$.  Using Lemma~\ref{CR3}, we only
  need to prove that every of its one step reducts is in $\llbracket C
  \rrbracket$.  If the reduction takes place in $t_1$ or $t_2$, then we
  apply Lemma~\ref{closure} and the induction hypothesis.

  Otherwise, the proof $t_1$ has the form $\pair{u}{v}$ and the
  reduct is $(u/x)t_2$.  As $\pair{u}{v} \in \llbracket A
  \with  B \rrbracket$, we have $u \in \llbracket A \rrbracket$.
  Hence, $(u/x)t_2 \in \llbracket C \rrbracket$.
\end{proof}

\begin{lemma}[Adequacy of $\elimwith^2$]
  \label{elimwith2}
  If $t_1 \in \llbracket A \with  B \rrbracket$ and,
  for all $u$ in $\llbracket B \rrbracket$,
  $(u/x)t_2 \in \llbracket C \rrbracket $, 
  then $\elimwith^2(t_1, \abstr{x}t_2) \in \llbracket C \rrbracket$.
\end{lemma}

\begin{proof}
  Similar to the proof of Lemma~\ref{elimwith1}.
\end{proof}

\begin{lemma}[Adequacy of $\elimplus$]
  \label{elimplus}
  If $t_1 \in \llbracket A \oplus B \rrbracket$, for all $u$ in $\llbracket A
  \rrbracket$, $(u/x)t_2 \in \llbracket C \rrbracket $, and, for all $v$
  in $\llbracket B \rrbracket$, $(v/y)t_3 \in \llbracket C \rrbracket $,
  then $\elimplus(t_1, \abstr{x}t_2, \abstr{y}t_3) \in \llbracket C \rrbracket$.
\end{lemma}

\begin{proof}
  By Lemma~\ref{Var}, $x \in \llbracket A \rrbracket$, thus $t_2 =
  (x/x)t_2 \in \llbracket C \rrbracket$. In the same way, $t_3 \in
  \llbracket C \rrbracket$.  Hence, $t_1$, $t_2$, and $t_3$ strongly
  terminate.  We prove, by induction on $\length{t_1} + \length{t_2} + \length{t_3}$,
  that $\elimplus(t_1, \abstr{x}t_2,
  \abstr{y}t_3) \in \llbracket C \rrbracket$.  Using Lemma~\ref{CR3}, we
  only need to prove that every of its one step reducts
  is in $\llbracket C \rrbracket$.  If the reduction takes place in
  $t_1$, $t_2$, or $t_3$, then we apply Lemma~\ref{closure} and
  the induction hypothesis. Otherwise, either:
  \begin{itemize}
    \item The proof $t_1$ has the form $\inl(w_2)$ and the reduct is
      $(w_2/x)t_2$. As $\inl(w_2) \in \llbracket A \oplus B \rrbracket$, we
      have $w_2 \in \llbracket A \rrbracket$.  Hence, $(w_2/x)t_2 \in
      \llbracket C \rrbracket$.

    \item The proof $t_1$ has the form $\inr(w_3)$ and the reduct is
      $(w_3/x)t_3$. As $\inr(w_3) \in \llbracket A \oplus B \rrbracket$, we
      have $w_3 \in \llbracket B \rrbracket$.  Hence, $(w_3/x)t_3 \in
      \llbracket C \rrbracket$.

    \item The proof $t_1$ has the form $t_1' \plus t''_1$ and the
      reduct is $\elimplus(t'_1, \abstr{x}t_2, \abstr{y}t_3) \plus
      \elimplus(t''_1, \abstr{x}t_2, \abstr{y}t_3)$. As $t_1
      \longrightarrow t'_1$ with an ultra-reduction rule, we have by
      Lemma~\ref{closure}, $t'_1 \in \llbracket A \oplus B
      \rrbracket$.  Similarly, $t''_1 \in \llbracket A \oplus B
      \rrbracket$.  Thus, by induction hypothesis, $\elimplus(t'_1,
      \abstr{x}t_2, \abstr{y}t_3) \in \llbracket A \oplus B \rrbracket$
      and $\elimplus(t''_1, \abstr{x}t_2, \abstr{y}t_3) \in \llbracket A
      \oplus B \rrbracket$.  We conclude with Lemma~\ref{sum}.

    \item The proof $t_1$ has the form $a \bullet t_1'$ and the
      reduct is $a \bullet \elimplus(t'_1, \abstr{x}t_2, \abstr{y}t_3)$. As $t_1
      \longrightarrow t'_1$ with an ultra-reduction rule, we have by
      Lemma~\ref{closure}, $t'_1 \in \llbracket A \oplus B
      \rrbracket$.  
      Thus, by induction hypothesis, $\elimplus(t'_1,
      \abstr{x}t_2, \abstr{y}t_3) \in \llbracket A \oplus B \rrbracket$.
      We conclude with Lemma~\ref{prod}.
      \qedhere
  \end{itemize}
\end{proof}

\begin{theorem}[Adequacy]
  Let $\Gamma\vdash t:A$ and
  $\sigma$ be a substitution mapping each variable $x:B\in\Gamma$ to an element
  of $\llbracket B \rrbracket$, then $\sigma t \in \llbracket A
  \rrbracket$.
\end{theorem}

\begin{proof}
  By induction on the structure of $t$.

  If $t$ is a variable, then, by definition of $\sigma$, $\sigma t \in
  \llbracket A \rrbracket$.  For the sixteen other proof constructors,
  we use the Lemmas~\ref{sum} to~\ref{elimplus}.  As all cases
  are similar, we just give a few examples.

  \begin{itemize}
    \item If $t = \pair{u}{v}$, where $u$ is a proof of $C$ and $v$ a proof of
      $D$, then, by induction hypothesis, $\sigma u \in \llbracket C
      \rrbracket$ and $\sigma v \in \llbracket D \rrbracket$.  Hence, by
      Lemma~\ref{pair}, $\pair{\sigma u}{\sigma v} \in \llbracket C \with  D
      \rrbracket$, that is $\sigma t \in \llbracket A \rrbracket$.

    \item If $t = \elimwith^1(u_1,\abstr{x}u_2)$, where
      $u_1$ is a proof of $C \with  D$ and $u_2$ a proof of $A$, then, by
      induction hypothesis, $\sigma u_1 \in
      \llbracket C \with  D \rrbracket$, for all $v$ in $\llbracket C
      \rrbracket$, $(v/x)\sigma u_2 \in \llbracket A \rrbracket$. Hence,
      by Lemma~\ref{elimwith1},
      $\elimwith^1(\sigma u_1,\abstr{x} \sigma u_2)
      \in \llbracket A \rrbracket$, that is $\sigma t \in \llbracket A
      \rrbracket$.
      \qedhere
  \end{itemize}
\end{proof}

\begin{corollary}[Termination]
  \label{termination}
  Let $\Gamma\vdash t:A$. Then,
  $t$ strongly terminates.
\end{corollary}

\begin{proof}
  Let $\sigma$ be the substitution mapping each variable $x:B\in\Gamma$ 
  to itself. Note that, by Lemma~\ref{Var}, this variable is an
  element of $\llbracket B \rrbracket$.  Then, $t = \sigma t$ is an
  element of $\llbracket A \rrbracket$. Hence, it strongly terminates.
\end{proof}

\subsection{Introduction}

We now prove that closed irreducible proofs end with an introduction rule,
except those of propositions of the form $A \otimes B$ and
$A \oplus B$, that are linear combinations of 
proofs ending with an introduction rule.

\begin{theorem}[Introduction]
  \label{introductions}
  Let $t$ be a closed irreducible proof of $A$.
  \begin{itemize}
    \item If $A$ has the form $\one$, then $t$ has the form $a.\star$.

    \item If $A$ has the form $B \multimap C$, then $t$ has the form
      $\lambda \abstr{x}u$.

    \item If $A$ has the form $B \otimes C$, then $t$ has the form $u \otimes v$,
      $u \plus v$, or $a \bullet u$.

    \item If $A$ has the form $\top$, then $t$ is $\topintro$.

    \item The proposition $A$ is not $\zero$.

    \item If $A$ has the form $B \with  C$, then $t$ has the form
      $\pair{u}{v}$.

    \item If $A$ has the form $B \oplus C$, then $t$ has the form $\inl(u)$,
      $\inr(u)$, $u \plus v$, or $a \bullet u$.
  \end{itemize}
\end{theorem}

\begin{proof}
  By induction on the structure of $t$.

  We first remark that, as the proof $t$ is closed, it is not a
  variable. Then, we prove that it cannot be an elimination.
  \begin{itemize}

    \item If $t = \elimone(u,v)$, then $u$ is a closed
      irreducible proof of $\one$, hence, by induction
      hypothesis, it has the form $a.\star$
      and the proof $t$ is reducible.

    \item If $t = u~v$, then $u$ is a closed  irreducible proof of $B
      \multimap A$, hence, by induction hypothesis, it has the form
      $\lambda \abstr{x}u_1$ and the proof $t$ is reducible.

    \item If $t = \elimtens(u,\abstr{xy}v)$, then $u$ is a closed
      irreducible proof of $B \otimes C$, hence, by induction hypothesis, it
      has the form $u_1 \otimes u_2$, $u_1 \plus u_2$, or $a \bullet
      u_1$ and the proof $t$ is reducible.

    \item If $t = \elimzero(u)$, then $u$ is a closed irreducible proof of
      $\zero$ and, by induction hypothesis, no such proof exists.

    \item If $t = \elimwith^1(u,\abstr{x}v)$, then $u$ is a closed
      irreducible proof of $B \with  C$, hence, by induction
      hypothesis, it has the form $\pair{u_1}{u_2}$
      and the proof $t$ is reducible.

    \item If $t = \elimwith^2(u,\abstr{x}v)$, then $u$ is a closed
      irreducible proof of $B \with  C$, hence, by induction
      hypothesis, it has the form $\pair{u_1}{u_2}$
      and the proof $t$ is reducible.

    \item If $t = \elimplus(u,\abstr{x}v,\abstr{y}w)$, then $u$ is a closed
      irreducible proof of $B \oplus C$, hence, by induction hypothesis, it
      has the form $\inl(u_1)$, $\inr(u_1)$, $u_1 \plus u_2$, or $a \bullet
      u_1$ and the proof $t$ is reducible.

  \end{itemize}

  Hence, $t$ is an introduction, a sum, or a product.

  It $t$ has the form $a.\star$, then $A$ is $\one$.  If it has the form
  $\lambda \abstr{x}u$, then $A$ has the form $B \multimap C$.
  If it
  has the form $u \otimes v$, then $A$ has the form $B \otimes C$.
  If it
  is $\topintro$, then $A$ is $\top$.
  If it
  has the form $\pair{u}{v}$, then $A$ has the form $B \with  C$.
  If it has
  the form $\inl(u)$ or $\inr(u)$, then $A$ has the form $B \oplus C$.
  We prove that, if it has the form $u \plus v$ or $a \bullet u$, $A$ has
  the form $B \otimes C$ or $B \oplus C$.

  \begin{itemize}
    \item 
      If $t = u \plus v$, then 
      the proofs $u$ and $v$ are two closed and irreducible proofs of
      $A$. If $A = \one$ then, by induction hypothesis, they both have the form
      $a.\star$
      and the proof $t$ is reducible.
      If $A$ has the form $B \multimap C$ then, by
      induction hypothesis, they are both abstractions and the proof $t$ is
      reducible.
      If $A = \top$ then, by induction hypothesis, they both are 
      $\topintro$
      and the proof $t$ is reducible.
      If $A = \zero$ then, they are
      irreducible proofs of $\zero$ and, by induction hypothesis, no such
      proofs exist.
      If $A$ has the form $B \with  C$, then, by induction
      hypothesis, they are both pairs and the proof $t$ is reducible.
      Hence, 
      $A$ has the form $B \otimes C$ or $B \oplus C$.

    \item
      If $t = a \bullet u$, then 
      the proofs $u$ is a closed and irreducible proof of
      $A$. If $A = \one$ then, by induction hypothesis, $u$ has the form $b.\star$
      and the proof $t$ is reducible.
      If $A$ has the form $B \multimap C$ then, by
      induction hypothesis, it is an abstraction and the proof $t$ is
      reducible.
      If $A = \top$ then, by induction hypothesis, it is $\topintro$
      and the proof $t$ is reducible.
      If $A = \zero$ then, it is 
      an irreducible proof of $\zero$ and, by induction hypothesis, no such
      proof exists.
      If $A$ has the form $B \with  C$, then, by induction
      hypothesis, it is a pair and the proof $t$ is reducible.
      Hence,
      $A$ has the form $B \otimes C$ or $B \oplus C$.
      \qedhere
  \end{itemize}
\end{proof}

\begin{remark}
  We reap here the benefit of commuting, when possible, the interstitial
  rules with the introduction rules, as closed irreducible proofs of
  $\one$, $A \multimap B$, $\top$ and $A \with B$ are genuine
  introductions.

  Those of $A \otimes B$ and $A \oplus B$ are linear combinations of
  introductions.  But $u \otimes v \plus u' \otimes v'$ is not
  convertible to $u' \otimes v' \plus u \otimes v$, $u \otimes v \plus u
  \otimes v$ is not convertible to $2 \bullet (u \otimes v)$, $u_1
  \otimes v \plus u_2 \otimes v$ is not convertible to $(u_1 \plus u_2)
  \otimes v$, etc. Thus, the proof of $A \otimes B$ are formal, rather
  than genuine, linear combinations of pairs formed with a closed
  irreducible proof of $A$ and a closed irreducible proof of $B$.  Such
  a set still need to be quotiented by a proper equivalence relation to
  provide the tensor product of the two semi-modules
  \cite{DiazcaroMalherbe22}.
\end{remark}

\begin{lemma}[Disjunction]
  If the proposition $A \oplus B$ has a closed proof, then $A$ has a
  closed proof or $B$ has a closed proof.
\end{lemma}

\begin{proof}
  Consider a closed proof of $A \oplus B$ and its irreducible form $t$.
  We prove, by induction on the structure of $t$, that $A$ has a
  closed proof or $B$ has a closed proof. By Theorem~\ref{introductions}, $t$
  has the form $\inl(u)$, $\inr(u)$, $u \plus
  v$, or $a \bullet u$.  If it has the form $\inl(u)$, $u$ is a closed
  proof of $A$. If it has the form $\inr(u)$, $u$ is a closed proof of
  $B$. If it has the form $u \plus v$ or $a \bullet u$, $u$ is a closed
  irreducible proof of $A \oplus B$. Thus, by induction hypothesis, $A$
  has a closed proof or $B$ has a closed proof.
\end{proof}

\section{Vectors and matrices}
\label{secvectorsmatrices}

From now on, we take the set of scalars $\mathcal S$ to be a field, rather than just a semi-ring, to aid the intuition with vector spaces. However, all the results are also valid for semi-modules over the semi-ring $\mathcal S$, except those considering the additive inverse (namely, Definition~\ref{def:additiveinverse} and item~\ref{vect:negative} of Lemma~\ref{vecstructure}).

\subsection{Vectors}
\label{secvectors}

As there is one rule $\one$-i for each scalar $a$, there is one closed
irreducible proof $a.\star$ for each scalar $a$.  Thus, the closed irreducible
proofs $a.\star$ of $\one$ are in one-to-one correspondence with the elements
of ${\mathcal S}$.  Therefore, the proofs $\pair{a.\star}{b.\star}$ of $\one
\with  \one$ are in one-to-one with the elements of ${\mathcal S}^2$, the
proofs $\pair{\pair{a.\star}{b.\star}}{c.\star}$ of $(\one \with  \one) \with 
\one$, and also the proofs $\pair{a.\star}{\pair{b.\star}{c.\star}}$ of $\one
\with  (\one \with  \one)$, are in one-to-one correspondence with the elements
of ${\mathcal S}^3$, etc.

Hence, as any vector space of finite dimension $n$ is isomorphic to
${\mathcal S}^n$, we have a way to express the vectors of any
${\mathcal S}$-vector space of finite dimension.  Yet, choosing an
isomorphism between a vector space and ${\mathcal S}^n$ amounts to
choosing a basis in this vector space, thus the expression of a vector
depends on the choice of a basis. This situation is analogous to that
of matrix formalisms. Matrices can represent vectors and linear
functions, but the matrix representation is restricted to finite
dimensional vector spaces, and the representation of a vector depends
on the choice of a basis. A change of basis in the vector space is
reflected by the use of a transformation matrix.

\begin{definition}[${\mathcal V}$]
  The set ${\mathcal V}$ is inductively defined as follows: $\one \in
  {\mathcal V}$, and if $A$ and $B$ are in ${\mathcal V}$, then so is $A
  \with  B$.
\end{definition}

We now show that if $A \in {\mathcal V}$, then the set of closed
irreducible proofs of $A$ has a vector space structure.

\begin{definition}[Zero vector]
  \label{def:zerovector}
  If $A \in {\mathcal V}$, we define the proof $0_A$ of $A$ by induction
  on $A$.  If $A = \one$, then $0_A = 0.\star$.  If $A = A_1 \with  A_2$,
  then $0_A = \pair{0_{A_1}}{0_{A_2}}$.
\end{definition}

\begin{definition}[Additive inverse]\label{def:additiveinverse}
  If $A \in {\mathcal V}$, and $t$ is a proof of $A$, we define the
  proof $- t$ of $A$ by induction on $A$.  If $A = \one$, then $t$
  reduces to $a.\star$, we let $- t = (-a).\star$. If $A = A_1 \with 
  A_2$, $t$ reduces to $\pair{t_1}{t_2}$ where $t_1$ is a proof of $A_1$
  and $t_2$ of $A_2$. We let $-t = \pair{- t_1}{- t_2}$.
\end{definition}

\begin{lemma} \label{vecstructure}
  If $A \in {\mathcal V}$ and $t$, $t_1$, $t_2$, and $t_3$ are closed proofs of
  $A$, then

  \begin{multicols}{2}
    \begin{enumerate}
      \item $(t_1 \plus t_2) \plus t_3 \equiv t_1 \plus (t_2 \plus t_3)$
      \item $t_1 \plus t_2 \equiv t_2 \plus t_1$
      \item $t \plus 0_A \equiv t$
      \item\label{vect:negative} $t \plus - t \equiv 0_A$
      \item $a \bullet b \bullet t \equiv (a \times b) \bullet t$
      \item $1 \bullet t \equiv t$
      \item $a \bullet (t_1 \plus t_2) \equiv a \bullet t_1 \plus a \bullet t_2$
      \item $(a + b) \bullet t \equiv a \bullet t \plus b \bullet t$
    \end{enumerate}
  \end{multicols}
\end{lemma}

\begin{proof}
  ~
  \begin{enumerate}
    \item By induction on $A$. If $A = \one$, then $t_1$, $t_2$, and $t_3$
      reduce respectively to $a.\star$, $b.\star$, and $c.\star$. We have
      $$(t_1 \plus t_2) \plus t_3 \lras ((a + b) + c).\star = (a + (b +
      c)).\star \llas t_1 \plus (t_2 \plus t_3)$$
      If $A = A_1 \with  A_2$, then $t_1$, $t_2$, and $t_3$ reduce
      respectively to $\pair{u_1}{v_1}$, $\pair{u_2}{v_2}$, and
      $\pair{u_3}{v_3}$.  Using the induction hypothesis, we have
      \begin{align*}
	(t_1 \plus t_2) \plus t_3 \lras &\ 
	\pair{(u_1 \plus u_2) \plus u_3}{(v_1 \plus v_2) \plus v_3}\\
	\equiv &\ 
	\pair{u_1 \plus (u_2 \plus u_3)}{v_1 \plus (v_2 \plus v_3)}
	\llas t_1 \plus (t_2 \plus t_3)
      \end{align*}

    \item By induction on $A$.  If $A = \one$, then $t_1$ and $t_2$ reduce
      respectively to $a.\star$ and $b.\star$. We have
      $$t_1 \plus t_2 \lras (a + b).\star = (b + a).\star \llas t_2 \plus t_1$$
      If $A = A_1 \with  A_2$, then $t_1$ and $t_2$ reduce respectively to
      $\pair{u_1}{v_1}$ and $\pair{u_2}{v_2}$.  Using the induction
      hypothesis, we have
      $$t_1 \plus t_2 \lras \pair{u_1 \plus u_2}{v_1 \plus v_2} \equiv
      \pair{u_2 \plus u_1}{v_2 \plus v_1} \llas t_2 \plus t_1$$

    \item By induction on $A$. If $A = \one$, then $t$ reduces to
      $a.\star$. We have
      $$t \plus 0_A \lras (a + 0).\star = a.\star \llas t$$
      If $A = A_1 \with  A_2$, then $t$ reduces to $\pair{u}{v}$.  Using
      the induction hypothesis, we have
      $$t \plus 0_A \lras \pair{u \plus 0_{A_1}}{v \plus 0_{A_2}}
      \equiv \pair{u}{v} \llas t$$

    \item By induction on $A$.  If $A = \one$, then $t$ reduces to
      $a.\star$.  We have
      $$t \plus -t \lras {a.\star} \plus (-a).\star \lra (a+(-a)).\star
      = 0.\star = 0_A$$
      If $A = A_1 \with  A_2$, then $t$ reduces to $\pair{u}{v}$.  Using
      the induction hypothesis, we have
      $$t \plus - t \lras \pair{u \plus - u}{v \plus - v} \equiv
      \pair{0_{A_1}}{0_{A_2}} = 0_A$$

    \item By induction on $A$.  If $A = \one$, then $t$ reduces to
      $c.\star$. We have
      $$a \bullet b \bullet t \lras (a \times (b \times c)).\star
      = ((a \times b) \times c).\star \llas (a \times b) \bullet t$$
      If $A = A_1 \with  A_2$, then $t$ reduces to $\pair{u}{v}$. Using
      the induction hypothesis, we have
      $$a \bullet b \bullet t \lras \pair{a \bullet b \bullet u}{a \bullet b
      \bullet v} \equiv \pair{(a \times b) \bullet u}{(a \times b) \bullet v}
      \llas (a \times b) \bullet t$$

    \item By induction on $A$.  If $A = \one$, then $t$ reduces to
      $a.\star$. We have
      $$1 \bullet t \lras(1 \times a).\star = a.\star \llas t$$
      If $A = A_1 \with  A_2$, then $t$ reduces to $\pair{u}{v}$. Using
      the induction hypothesis, we have
      $$1 \bullet t \lras \pair{1 \bullet u}{1 \bullet v} \equiv \pair{u}{v}
      \llas t$$

    \item By induction on $A$.  If $A = \one$, then $t_1$ and $t_2$ reduce
      respectively to $b.\star$ and $c.\star$. We have
      $$a \bullet (t_1 \plus t_2) \lras (a \times (b + c)).\star
      = (a \times b + a \times c).\star \llas  a \bullet t_1 \plus a \bullet t_2$$
      If $A = A_1 \with  A_2$, then $t_1$ and $t_2$ reduce respectively to
      $\pair{u_1}{v_1}$ and $\pair{u_2}{v_2}$.  Using the induction
      hypothesis, we have
      \begin{align*}
      a \bullet (t_1 \plus t_2) \lras &\ \pair{a \bullet (u_1 \plus u_2)}{a
      \bullet (v_1 \plus v_2)}\\
      \equiv &\ \pair{a \bullet u_1 \plus a
      \bullet u_2}{a \bullet v_1 \plus a \bullet v_2} \llas a \bullet
      t_1 \plus a \bullet t_2
      \end{align*}

    \item By induction on $A$.  If $A = \one$, then $t$ reduces to
      $c.\star$. We have
      $$(a + b) \bullet t \lras ((a + b) \times c).\star =
      (a \times c + b \times c).\star \llas a \bullet t \plus b \bullet t$$
      If $A = A_1 \with  A_2$, then $t$ reduces to $\pair{u}{v}$.  Using
      the induction hypothesis, we have
      \begin{align*}
        (a + b) \bullet t 
        \lras~&\pair{(a + b) \bullet u}{(a + b) \bullet v}\\
        \equiv~& \pair{a \bullet u \plus b \bullet u}{a \bullet v \plus b \bullet v}
        \llas a \bullet t \plus b \bullet t
        \tag*{\qedhere}
      \end{align*}
  \end{enumerate}
\end{proof}

\begin{definition}[Dimension of a proposition in ${\mathcal V}$]
  To each proposition $A \in {\mathcal V}$, we associate a positive
  natural number $d(A)$, which is the number of occurrences of the
  symbol $\one$ in $A$: $d(\one) = 1$ and $d(B
  \with  C) = d(B) + d(C)$.
\end{definition}

If $A \in {\mathcal V}$ and $d(A) = n$, then the closed irreducible proofs
of $A$ and the vectors of ${\mathcal S}^n$ are in one-to-one
correspondence: to each closed irreducible proof $t$ of $A$, we
associate a vector $\underline{t}$ of ${\mathcal S}^n$ and to each
vector ${\bf u}$ of ${\mathcal S}^n$, we associate a closed
irreducible proof $\overline{\bf u}^A$ of $A$.

\begin{definition}[One-to-one correspondence]
  \label{onetoone}
  Let $A \in {\mathcal V}$ with $d(A) = n$.  To each closed irreducible
  proof $t$ of $A$, we associate a vector $\underline{t}$ of ${\mathcal
  S}^n$ as follows.
  \begin{itemize}
    \item
      If $A = \one$, then $t = a.\star$. We let $\underline{t} =
      \left(\begin{smallmatrix} a \end{smallmatrix}\right)$.

    \item
      If $A = A_1 \with  A_2$, then $t = \pair{u}{v}$.  We let
      $\underline{t}$ be the vector with two blocks $\underline{u}$ and
      $\underline{v}$: $\underline{t} = \left(\begin{smallmatrix}
        \underline{u}\\\underline{v} \end{smallmatrix}\right)$.
      Remind that, using the block notation, if ${\bf u} =
      \left(\begin{smallmatrix} 1\\2 \end{smallmatrix}\right)$
      and ${\bf v} =
      \left(\begin{smallmatrix} 3 \end{smallmatrix}\right)$,
      then 
      $\left(\begin{smallmatrix}
        {\bf u}\\{\bf v} \end{smallmatrix}\right) =
        \left(\begin{smallmatrix}
        1\\2\\3 \end{smallmatrix}\right)$ and not
       $\left(\begin{smallmatrix}
          \left(\begin{smallmatrix} 1\\2 \end{smallmatrix}\right)
          \\
          \left(\begin{smallmatrix} 3 \end{smallmatrix}\right)
        \end{smallmatrix}\right)$.
  \end{itemize}

  To each vector ${\bf u}$ of ${\mathcal S}^n$, we associate a closed
  irreducible proof $\overline{\bf u}^A$ of $A$.

  \begin{itemize}
    \item If $n = 1$, then ${\bf u} = \left(\begin{smallmatrix}
      a \end{smallmatrix}\right)$. We let $\overline{\bf u}^A = a.\star$.

    \item If $n > 1$, then $A = A_1 \with  A_2$, let $n_1$ and $n_2$ be
      the dimensions of $A_1$ and $A_2$.  Let ${\bf u}_1$ and ${\bf u}_2$
      be the two blocks of ${\bf u}$ of $n_1$ and $n_2$ lines, so ${\bf u}
      = \left(\begin{smallmatrix} {\bf u}_1\\ {\bf
      u}_2\end{smallmatrix}\right)$.  We let $\overline{\bf u}^A =
      \pair{\overline{{\bf u}_1}^{A_1}}{\overline{{\bf u}_2}^{A_2}}$.
  \end{itemize}
\end{definition}

We extend the definition of $\underline{t}$ to any closed proof of
$A$, $\underline{t}$ is by definition $\underline{t'}$ where $t'$ is
the irreducible form of $t$.

The next lemmas show that the symbol $\plus$ expresses the sum of
vectors and the symbol $\bullet$, the product of a vector by a scalar.

\begin{lemma}[Sum of two vectors]
  \label{parallelsum}
  Let $A \in {\mathcal V}$, and $u$ and $v$ be two closed proofs of $A$.
  Then, $\underline{u \plus v} = \underline{u} + \underline{v}$.
\end{lemma}

\begin{proof}
  By induction on $A$.

  \begin{itemize}
    \item  If $A = \one$, then $u \lras a.\star$, $v \lras b.\star$,
      $\underline{u} =
      \left(\begin{smallmatrix} a \end{smallmatrix}\right)$,
      $\underline{v} =
      \left(\begin{smallmatrix} b \end{smallmatrix}\right)$.
      Thus, $\underline{u \plus v} = \underline{{a.\star} \plus {b.\star}} =
      \underline{(a + b).\star}
      = \left(\begin{smallmatrix} a + b \end{smallmatrix}\right)
      = \left(\begin{smallmatrix} a \end{smallmatrix}\right)
      + \left(\begin{smallmatrix} b \end{smallmatrix}\right)
      = \underline {u} + \underline{v}$.

    \item 
      If $A = A_1 \with  A_2$, then $u \lras \pair{u_1}{u_2}$, $v \lra
      \pair{v_1}{v_2}$, $\underline{u} = \left(\begin{smallmatrix}
	  \underline{u_1} \\ \underline{u_2}
      \end{smallmatrix}\right)$ and $\underline{v} = \left(\begin{smallmatrix} \underline{v_1} \\ \underline{v_2} \end{smallmatrix}\right)$.
      Thus, using the induction hypothesis, $\underline{u \plus v} =
      \underline{\pair{u_1}{u_2} \plus \pair{v_1}{v_2}} =
      \underline{\pair{u_1 \plus v_1}{u_2 \plus v_2}} =
      \left(\begin{smallmatrix} \underline{u_1 \plus v_1} \\ \underline{u_2
	  \plus v_2}
      \end{smallmatrix}\right)
      = 
      \left(\begin{smallmatrix} \underline{u_1} + \underline{v_1} \\
	  \underline{u_2} + \underline{v_2}
      \end{smallmatrix}\right)
      =
      \left(\begin{smallmatrix} \underline{u_1} \\
	  \underline{u_2}
      \end{smallmatrix}\right)
      +
      \left(\begin{smallmatrix} \underline{v_1} \\
	  \underline{v_2}
      \end{smallmatrix}\right) = 
      \underline{u} + \underline{v}$.
      \qedhere
  \end{itemize}
\end{proof}

\begin{lemma}[Product of a vector by a scalar]
  \label{parallelprod}
  Let $A \in {\mathcal V}$ and $u$ be a closed proof of $A$.  Then,
  $\underline{a \bullet u} = a \underline{u}$.
\end{lemma}

\begin{proof}
  By induction on $A$. 

  \begin{itemize}
    \item If $A = \one$, then $u \lras b.\star$, $\underline{u} =
      \left(\begin{smallmatrix} b \end{smallmatrix}\right)$, Thus
      $\underline{a \bullet u} = \underline{a \bullet b.\star} =
      \underline{(a \times b).\star} = \left(\begin{smallmatrix} a \times
      b \end{smallmatrix}\right) = a \left(\begin{smallmatrix}
      b \end{smallmatrix}\right) = a \underline {u}$.

    \item If $A = A_1 \with  A_2$, then $u \lras \pair{u_1}{u_2}$,
      $\underline{u} =
      \left(\begin{smallmatrix} \underline{u_1} \\ \underline{u_2}
      \end{smallmatrix}\right)$.
      Thus, using the induction hypothesis, $\underline{a \bullet u} =
      \underline{a \bullet \pair{u_1}{u_2}} = \underline{\pair{a \bullet
      u_1}{a \bullet u_2}} = \left(\begin{smallmatrix} \underline{a
	  \bullet u_1} \\ \underline{a \bullet u_2}
      \end{smallmatrix}\right)
      = 
      \left(\begin{smallmatrix} a \underline{u_1}  \\
	  a \underline{u_2} 
      \end{smallmatrix}\right)
      =
      a \left(\begin{smallmatrix} \underline{u_1} \\
	  \underline{u_2}
      \end{smallmatrix}\right)
      = a \underline{u}$
      \qedhere
  \end{itemize}
\end{proof}

\begin{remark}
  We have seen that the rules
  \[
    \begin{array}{r@{\,}l@{\qquad\qquad}r@{\,}l}
      {a.\star} \plus b.\star & \longrightarrow  (a+b).\star
      &
      a \bullet b.\star & \longrightarrow  (a \times b).\star\\
      \pair{t}{u} \plus \pair{v}{w} & \longrightarrow  \pair{t \plus v}{u \plus w}
      &
      a \bullet \pair{t}{u} & \longrightarrow  \pair{a \bullet t}{a \bullet u}
    \end{array}
  \]
  are commutation rules between the interstitial rules, sum and prod,
  and introduction rules $\one$-i and $\with $-i.

  Now, these rules appear to be also vector calculation rules.
\end{remark}

\subsection{Matrices}
\label{secmatrices}

We now want to prove that if $A, B \in {\mathcal V}$ with $d(A) = m$
and $d(B) = n$, and $F$ is a linear function from ${\mathcal S}^m$ to
${\mathcal S}^n$, then there exists a closed proof $f$ of $A
\multimap B$ such that, for all vectors ${\bf u} \in {\mathcal S}^m$,
$\underline{f~\overline{\bf u}^A} = F({\bf u})$.  This can
equivalently be formulated as the fact that if $M$ is a matrix with
$m$ columns and $n$ lines, then there exists a closed proof $f$ of $A
\multimap B$ such that for all vectors ${\bf u} \in {\mathcal S}^m$,
$\underline{f~\overline{\bf u}^A} = M {\bf u}$.

A similar theorem has been proved in \cite{odot} for a non-linear
calculus. The proof of the following theorem is just a check that the
construction given there verifies the linearity constraints of the
${\mathcal L}^{\mathcal S}$-calculus.

\begin{theorem}[Matrices]
  \label{matrices}
  Let $A, B \in {\mathcal V}$ with $d(A) = m$ and $d(B) = n$ and let $M$
  be a matrix with $m$ columns and $n$ lines, then there exists a closed
  proof $t$ of $A \multimap B$ such that, for all vectors ${\bf u} \in
  {\mathcal S}^m$, $\underline{t~\overline{\bf u}^A} = M {\bf u}$.
\end{theorem}

\begin{proof}
  By induction on $A$.
  \begin{itemize}

    \item If $A = \one$, then $M$ is a matrix of one column and
      $n$ lines. Hence, it is also a vector of $n$ lines.
      We take
      $$t = \lambda \abstr{x} \elimone(x,\overline{M}^B)$$
      Let ${\bf u} \in {\mathcal S}^1$, ${\bf u}$ has the form
      $\left(\begin{smallmatrix}  a \end{smallmatrix}\right)$ and 
      $\overline{\bf u}^A = a.\star$.

      Then, using
      Lemma~\ref{parallelprod}, we have $\underline{t~\overline{\bf u}^A} 
      = \underline{\elimone(\overline{\bf u}^A,\overline{M}^B)}
      = \underline{\elimone(a.\star,\overline{M}^B)}
      = \underline{a \bullet \overline{M}^B}
      = a \underline{\overline{M}^B} = a M = M
      \left(\begin{smallmatrix}  a\end{smallmatrix}\right) =
      M {\bf u}$.

    \item If $A = A_1 \with  A_2$, then let $d(A_1) = m_1$
      and $d(A_2) = m_2$. Let $M_1$ and $M_2$ be the two
      blocks of $M$ of $m_1$ and $m_2$ columns, so $M =
      \left(\begin{smallmatrix} M_1 & M_2\end{smallmatrix}\right)$.

      By induction hypothesis, there exist closed proofs $t_1$ and $t_2$ of
      the propositions $A_1 \multimap B$ and $A_2 \multimap B$ such
      that, for all vectors ${\bf u}_1 \in {\mathcal S}^{m_1}$ and ${\bf
      u}_2 \in {\mathcal S}^{m_2}$, we have $\underline{t_1~\overline{{\bf
      u}_1}^{A_1}} = M_1 {\bf u}_1$ and $\underline{t_2~\overline{{\bf
      u}_2}^{A_ 2}} = M_2 {\bf u}_2$.  We take
      $$t = \lambda \abstr{x} (\elimwith^1(x,\abstr{y} (t_1~y)) \plus \elimwith^2(x, \abstr{z} (t_2~z)))$$
      Let ${\bf u} \in {\mathcal S}^m$, and ${\bf u}_1$ and ${\bf u}_2$ be
      the two blocks of $m_1$ and $m_2$ lines of ${\bf u}$, so ${\bf u} =
      \left(\begin{smallmatrix} {\bf u}_1 \\ {\bf
      u}_2 \end{smallmatrix}\right)$, and $\overline{\bf u}^A =
      \pair{\overline{{\bf u}_1}^{A_1}}{\overline{{\bf u}_2}^{A_ 2}}$.

      Then, using Lemma~\ref{parallelsum}, $\underline{t~\overline{\bf
      u}^A} = \underline{\elimwith^1(\pair{\overline{{\bf
	  u}_1}^{A_1}}{\overline{{\bf u}_2}^{A_ 2}}, \abstr{y}
	(t_1~y)) \plus \elimwith^2(\pair{\overline{{\bf
	  u}_1}^{A_1}}{\overline{{\bf u}_2}^{A_ 2}}, \abstr{z}
      (t_2~z))}
      = \underline{(t_1~\overline{{\bf u}_1}^{A_1}) \plus (t_2~\overline{{\bf u}_2}^{A_ 2})}
      = \underline{t_1~\overline{{\bf u}_1}^{A_1}} + \underline{t_2~\overline{{\bf u}_2}^{A_ 2}}
      = M_1 {\bf u}_1 + M_2 {\bf u}_2
      = \left(\begin{smallmatrix} M_1 & M_2 \end{smallmatrix}\right)
      \left(\begin{smallmatrix} {\bf u}_1 \\ {\bf u}_2  \end{smallmatrix}\right)
      = M {\bf u}$.
      \qedhere
  \end{itemize}
\end{proof}

\begin{remark}
  In the proofs $\elimone(x,\overline{M}^B)$, $\elimwith^1(x,\abstr{y}
  (t_1~y))$, and $\elimwith^2(x,\abstr{z} (t_2~z))$, the variable $x$ occurs in
  one argument of the symbols $\elimone$, $\elimwith^1$, and $\elimwith^2$, but
  not in the other.  In contrast, in the proof $\elimwith^1(x,\abstr{y}
  (t_1~y)) \plus \elimwith^2(x, \abstr{z} (t_2~z))$, it occurs in both
  arguments of the symbol $\plus$.  Thus, these proofs are well-formed proofs in the
  system of Figure~\ref{figuretypingrules}.
\end{remark}

\begin{remark}
  The rules
  \begin{align*}
    \elimone(a.\star,t) & \longrightarrow  a \bullet t
    & \elimwith^1(\pair{t}{u}, \abstr{x}v) & \longrightarrow  (t/x)v \\
    (\lambda \abstr{x}t)~u  & \longrightarrow  (u/x)t
    &
    \elimwith^2(\pair{t}{u}, \abstr{x}v) & \longrightarrow  (u/x)v
  \end{align*}
  were introduced as cut reduction rules.

  Now, these rules appear to be also matrix calculation rules.
\end{remark}

\begin{example}[Matrices with two columns and two lines]
  The matrix
  $\left(\begin{smallmatrix} a & c\\b & d \end{smallmatrix}\right)$
  is expressed as the proof
  $$t = \lambda \abstr{x} \elimwith^1(x,\abstr{y}
  \elimone(y,\pair{a.\star}{b.\star})) \plus \elimwith^2(x,\abstr{z}
  \elimone(z,\pair{c.\star}{d.\star}))$$

  Then
  \begin{align*}
    t~\pair{e.\star}{f.\star} & \lra
    \elimwith^1(\pair{e.\star}{f.\star},\abstr{y} \elimone(y,\pair{a.\star}{b.\star}))
    \plus
    \elimwith^2(\pair{e.\star}{f.\star},\abstr{z} \elimone(z,\pair{c.\star}{d.\star}))\\
    & \lras \elimone(e.\star,\pair{a.\star}{b.\star}) \plus
    \elimone(f.\star,\pair{c.\star}{d.\star})\\
    & \lras e \bullet \pair{a.\star}{b.\star} \plus
    f \bullet \pair{c.\star}{d.\star}\\
    & \lras
    \pair{(a \times e).\star}{(b \times e).\star}
    \plus 
    \pair{(c \times f).\star}{(d \times f).\star}\\
    & \lras
    \pair{(a \times e + c \times f).\star}{(b \times e + d \times f).\star}
  \end{align*}
\end{example}

\section{Linearity}
\label{seclinearity}

\subsection{Observational equivalence}\label{sec:obseq}

We now prove the converse: if $A, B \in {\mathcal V}$, then each closed
proof $t$ of $A \multimap B$ expresses a linear function, that is
that if $u_1$ and $u_2$ are closed proofs of $A$, then
$$t~(u_1 \plus u_2) \equiv t~u_1 \plus t~u_2
\qquad\qquad\textrm{and}\qquad\qquad
t~(a \bullet u_1) \equiv a \bullet t~u_1$$
The fact that we want all proofs of $\one \multimap \one$ to be
linear functions from ${\mathcal S}$ to ${\mathcal S}$ explains why
the rule sum must be additive. If it were multiplicative, the
proposition $\one \multimap \one$ would have the proof $g =
\lambda \abstr{x}{{x} \plus {1.\star}}$ that is not linear as
$g~({1.\star} \plus {1.\star}) \lras 3.\star \not\equiv 4.\star
\llas {(g~1.\star)} \plus {(g~1.\star)}$.

In fact, instead of proving that for all closed proofs $t$ of $A
\multimap B$
$$t~(u_1 \plus u_2) \equiv t~u_1 \plus t~u_2
\qquad\qquad\textrm{and}\qquad\qquad t~(a \bullet u_1) \equiv a
\bullet t~u_1$$ it is more convenient to prove the equivalent
statement that for a proof $t$ of $B$ such that $x:A \vdash t:B$
$$t\{u_1 \plus u_2\} \equiv t\{u_1\} \plus t\{u_2\}
\qquad\qquad\textrm{and}\qquad\qquad
t\{a \bullet u_1\} \equiv a \bullet t\{u_1\}$$

We can attempt to generalize this statement and prove that
these properties hold for all proofs, whatever the
proved proposition. But this generalization is too strong for two reasons.
First, as we have seen, the sum rule does not commute with the
introduction rules of $\otimes$ and $\oplus$. Thus, if, for example,
$A = \one$, $B = \one \oplus \one$, and $t = \inl(x)$, we have
$$t\{{1.\star} \plus 1.\star\} \lra \inl({2.\star})
\qquad\qquad\textrm{and}\qquad\qquad t\{1.\star\} \plus t\{1.\star\}
= \inl(1.\star) \plus \inl(1.\star)$$ and these two irreducible proofs
are different.  Second, the introduction rule of $\multimap$
introduces a free variable. Thus, if, for example, $A = \one$, $B =
(\one \multimap \one) \multimap \one$, and $t = \lambda \abstr{y} y~x$,
we have 
$$t~\{{1.\star} \plus 2.\star\} \lra \lambda \abstr{y} y~3.{\star}
\qquad\qquad\textrm{and}\qquad\qquad t\{1.\star\} \plus t\{2.\star\} \lra
\lambda \abstr{y} {(y~1.\star)} \plus {(y~2.\star)}$$
and these two irreducible proofs are different.

Note however that, although the proofs $\inl({2.\star})$ and
$\inl(1.\star) \plus \inl(1.\star)$ of $\one \oplus \one$ are
different, if we put them in the context
$\elimplus(\_,\abstr{x}x,\abstr{y}y)$, then both proofs
$\elimplus(\inl({1.\star} \plus 1.\star),\abstr{x}x,\abstr{y}y)$ and
$\elimplus(\inl(1.\star) \plus \inl(1.\star),\abstr{x}x,\abstr{y}y)$
reduce to $2.\star$. In the same way, although the proofs $\lambda
\abstr{y} y~3.\star$ and $\lambda \abstr{y} {(y~1.\star)} \plus
{(y~2.\star)}$ are different, if we put them in the context
$\_~\lambda \abstr{z} z$, then both proofs $(\lambda \abstr{y}
y~3.\star)~\lambda \abstr{z} z$ and $(\lambda \abstr{y}
{(y~1.\star)} \plus {(y~2.\star)})~\lambda \abstr{z} z$ reduce
to $3.\star$.  This leads us to introduce a notion of observational
equivalence.

\begin{definition}[Observational equivalence]\label{def:obseq}
  Two proofs $t_1$ and $t_2$ of a proposition $B$ are {\it
  observationally equivalent}, $t_1 \sim t_2$, if for all propositions
  $C$ in $\mathcal{V}$ and for all proofs $c$ such that $\_:B \vdash
  c:C$, we have
  $$c\{t_1\} \equiv c\{t_2\}$$
\end{definition}

We want to prove that for all proofs $t$ such that $x:A \vdash t:B$
and for all closed proofs $u_1$ and $u_2$ of $A$, we have
$$t\{u_1 \plus u_2\} \sim t\{u_1\} \plus t\{u_2\}
\qquad\qquad\textrm{and}\qquad\qquad
t\{a\bullet u_1\} \sim a \bullet t\{u_1\}$$

However, a proof of this property by induction on $t$ does not go
through and, to prove it, we first prove, in Theorem \ref{linearity},
that for all proofs $t$ such that $x:A \vdash t:B$, $B \in
\mathcal{V}$, for all closed proofs $u_1$ and $u_2$ of $A$
$$t\{u_1 \plus u_2\} \equiv t\{u_1\} \plus t\{u_2\}
\qquad\qquad\textrm{and}\qquad\qquad t\{a\bullet u_1\} \equiv a
\bullet t\{u_1\}$$ and we deduce the result for observational
equivalence in Corollary \ref{corollary0}.

\subsection{Measure of a proof}

The proof of Theorem \ref{linearity} proceeds by induction on the
measure of the proof $t$, and the first part of this proof is the
definition of such a measure function $\mu$.  Our goal could be to
build a measure function such that if $t$ is proof of $B$ in a context
$\Gamma, x:A$ and $u$ is a proof of $A$, then $\mu((u/x)t) = \mu(t) +
\mu(u)$. This would be the case, for the usual notion of size, if $x$
had exactly one occurrence in $t$. But, due to additive connectives,
the variable $x$ may have zero, one, or several occurrences in $t$.

First, as the rule $\zero$-e is additive, it may happen that
$\elimzero(t)$ is a proof in the context $\Gamma, x:A$, and $x$ has no
occurrence in $t$.  Thus, we lower our expectations to $\mu((u/x)t)
\leq \mu(t) + \mu(u)$, which is sufficient to prove the linearity
theorem.

Then, as the rules $\plus$, $\with $-i, and $\oplus$-e rules are
additive, if $u \plus v$ is proof of $B$ in a context $\Gamma, {x:A}$,
$x$ may occur both in $u$ and in $v$. And the same holds for the
proofs $\pair{u}{v}$, and $\elimplus(t,\abstr{x}u,\abstr{y}v)$. In these
cases, we modify the definition of the measure function and take $\mu(t
\plus u) = 1 + \max(\mu(t), \mu(u))$, instead of $\mu(t \plus u) = 1 +
\mu(t) + \mu(u)$, etc., making the function $\mu$ a mix between a size
function and a depth function.  
This leads to the following definition.

\begin{definition}[Measure of a proof]
  \label{measureofaproof}~
  \begin{align*}
    \mu(x) &= 0 &
    \mu(\topintro) &= 1\\
    \mu(t \plus u) &= 1 + \max(\mu(t), \mu(u)) & 
    \mu(\elimzero(t)) &= 1 + \mu(t)\\
    \mu(a \bullet t) &= 1 + \mu(t) & 
    \mu(\pair{t}{u}) &= 1 + \max(\mu(t), \mu(u))\\
    \mu(a.\star) &= 1 &
    \mu(\elimwith^1(t,\abstr{y}u)) &= 1 + \mu(t) + \mu(u)\\
    \mu(\elimone(t,u)) &= 1 + \mu(t) + \mu(u) &
    \mu(\elimwith^2(t,\abstr{y}u)) &= 1 + \mu(t) + \mu(u)\\
    \mu(\lambda \abstr{x} t) &= 1 + \mu(t) &
    \mu(\inl(t)) &= 1 + \mu(t)\\
    \mu(t~u) &= 1 + \mu(t) + \mu(u) &
    \mu(\inr(t)) &= 1 + \mu(t)\\
    \mu(t \otimes u) &= 1 + \mu(t) + \mu(u) &
    \mu(\elimplus(t,\abstr{y}u,\abstr{z}v)) &= 1 + \mu(t) + \max(\mu(u), \mu(v)) \\
    \mu(\elimtens(t,\abstr{xy} u)) &= 1 + \mu(t) + \mu(u) & 
  \end{align*}
\end{definition}

\begin{lemma}\label{lem:msubst}
If $\Gamma, x:A \vdash t:B$ and $\Delta \vdash u:A$ then
  $\mu((u/x)t) \leq \mu(t)+\mu(u)$.
\end{lemma}

\begin{proof}
  By induction on $t$.
  \begin{itemize}
    \item If $t$ is a variable, then $\Gamma$ is empty, $t = x$, 
      $(u/x)t = u$ and $\mu(t) = 0$.
      Thus, $\mu((u/x)t) = \mu(u) = \mu(t)+\mu(u)$.

    \item If $t = t_1 \plus t_2$, then $\Gamma, x:A \vdash t_1:B$,
      $\Gamma, x:A \vdash t_2:B$.  Using the induction hypothesis, we get 
      $\mu((u/x)t)
      = 1 +
      \max(\mu((u/x)t_1),\mu((u/x)t_2)) \leq 1 + \max(\mu(t_1) + \mu(u),
      \mu(t_2) + \mu(u))
      = \mu(t) + \mu(u)$.

    \item If $t = a \bullet t_1$, then $\Gamma, x:A \vdash t_1:B$.
      Using the induction hypothesis, we get 
      $\mu((u/x)t)
      = 1 + \mu((u/x)t_1)
      \leq 1 + \mu(t_1) + \mu(u) = \mu(t) + \mu(u)$.

    \item The proof $t$ cannot be of the form
      $a.\star$, that is not a proof in $\Gamma, x:A$.

      \item If $t = \elimone(t_1,t_2)$, then
      $\Gamma = \Gamma_1, \Gamma_2$ and there are two cases.
      \begin{itemize}
	\item 
	  If $\Gamma_1, x:A \vdash t_1:\one$ and $\Gamma_2 \vdash t_2:B$,
	  then, using the induction hypothesis, we get
	  $\mu((u/x)t)
	  = 1 + \mu((u/x)t_1) + \mu(t_2)
	  \leq 1 + \mu(t_1) + \mu(u) + \mu(t_2)
	  = \mu(t) + \mu(u)$.

	\item 
	  If $\Gamma_1 \vdash t_1:\one$ and $\Gamma_2, x:A \vdash t_2:B$,
	  then, using the induction hypothesis, we get
	  $\mu((u/x)t)
	  = 1 + \mu(t_1) + \mu((u/x)t_2)
	  \leq 1 + \mu(t_1) + \mu(t_2) + \mu(u)
	  = \mu(t) + \mu(u)$.
      \end{itemize}

    \item If $t = \lambda \abstr{y} t_1$,
we apply the same method as for the case $t = a \bullet t_1$.

    \item If $t = t_1~t_2$, we apply the same method as for the
      case $t = \elimone(t_1,t_2)$.

    \item If $t = t_1 \otimes t_2$, we apply the same method as for the
      case $t = \elimone(t_1,t_2)$.

    \item If $t = \elimtens(t_1, \abstr{yz} t_2)$,
      we apply the same method as for the
      case $t = \elimone(t_1,t_2)$.

    \item If $t = \topintro$, then
      $\mu((u/x)t)
      = 1
      \leq 1 + \mu(u) 
      = \mu(t) + \mu(u)$.

    \item If $t = \elimzero(t_1)$, then
      $\Gamma = \Gamma_1, \Gamma_2$ and there are two cases.
      \begin{itemize}
	\item 
	  If $\Gamma_1, x:A \vdash t_1:\zero$,
we apply the same method as for the case $t = a \bullet t_1$.

	\item 
	  If $\Gamma_1 \vdash t_1:\zero$, then, we get $\mu((u/x)t) = \mu(t)
	  \leq \mu(t) + \mu(u)$.
      \end{itemize}

    \item If $t = \pair{t_1}{t_2}$, we apply the same method as for the
      case $t = t_1 \plus t_2$.

    \item   
      If $t = \elimwith^1(t_1,\abstr{y}t_2)$, we apply the same method
      as for the case $t = \elimone(t_1,t_2)$.

    \item   
      If $t = \elimwith^2(t_1,\abstr{y}t_2)$, 
      we apply the same method as for the
      case $t = \elimone(t_1,t_2)$.

    \item If $t = \inl(t_1)$ or $t = \inr(t_1)$,
     we apply the same method as for the case $t = a \bullet t_1$.

    \item   
      If $t = \elimplus(t_1,\abstr{y}t_2,\abstr{z}t_3)$ then
      $\Gamma = \Gamma_1, \Gamma_2$ and there are two cases.
      \begin{itemize}
	\item 
	  If $\Gamma_1, x:A \vdash t_1:C_1 \oplus C_2$,
	  $\Gamma_2, y:C_1 \vdash t_2:A$,
	  $\Gamma_2, z:C_2 \vdash t_3:A$, then using the induction hypothesis,
	  we get
	  $\mu((u/x)t)
	  = 1 + \mu((u/x)t_1) + \max(\mu(t_2),\mu(t_3))
	  \leq 
	  1 + \mu(t_1) + \mu(u) + \max(\mu(t_2),\mu(t_3))
	  = \mu(t) + \mu(u)$.

	\item 
	  If $\Gamma_1 \vdash t_1:C_1 \oplus C_2$,
	  $\Gamma_2, y:C_1, x:A \vdash t_2:A$,
	  $\Gamma_2, z:C_2, x:A \vdash t_3:A$, then using the induction hypothesis,
	  we get
	  $\mu((u/x)t)
	  = 1 + \mu(t_1) + \max(\mu((u/x)t_2),\mu((u/x)t_3))
	  \leq 
	  1 + \mu(t_1) + \max(\mu(t_2) + \mu(u),\mu(t_3) + \mu(u))
	  =
	  1 + \mu(t_1) + \max(\mu(t_2),\mu(t_3)) + \mu(u)
	  = \mu(t) + \mu(u)$.
	  \qedhere
      \end{itemize}      
  \end{itemize}
\end{proof}

\begin{example}
  \label{inferieurstrict}
  Let $t = \elimzero(y)$ and $u = 1.\star$.
  We have $y:\zero, x:\one \vdash t:C$, $\mu(t) = 1$, $\mu(u) = 1$ and
  $\mu((u/x)t) = 1$. Thus,
  $\mu((u/x)t) \leq \mu(t) + \mu(u)$. 
\end{example}

As a corollary, we get a similar measure preservation theorem for
reduction.

\begin{lemma}
  \label{lem:mured}
  If $\Gamma \vdash t:A$ and $t \lra u$, then $\mu(t) \geq \mu(u)$.
\end{lemma}

\begin{proof}
  By induction on $t$. The context cases are trivial because the
  functions used to define $\mu(t)$ in function of $\mu$ of the
  subterms of $t$ are monotone. We check the rules one by one, using
  Lemma~\ref{lem:msubst}.

  \begin{itemize}
    \item $\mu(\elimone(a.\star,t)) = 2 + \mu(t) > 1 + \mu(t) = \mu(a
      \bullet t)$
    \item $\mu((\lambda \abstr{x}t)~u) = 2 + \mu(t) + \mu(u) > \mu(t) +
      \mu(u) \geq \mu((u/x)t)$
    \item $\mu(\elimtens(u \otimes v,\abstr{x y}w))
      = 2 + \mu(u) + \mu(v) + \mu(w)
      > \mu(u) + \mu(v) + \mu(w)
      \geq \mu(u) + \mu((v/y)w)
      \geq  \mu((u/x)(v/y)w)
      = \mu((u/x,v/y)w)$
      as $x$ does not occur in $v$
    \item $\mu(\elimwith^1(\pair{t}{u}, \abstr{x}v)) = 2 +
      \max(\mu(t),\mu(u)) + \mu(v) > \mu(t) + \mu(v) \geq \mu((t/x)v)$
    \item $\mu(\elimwith^2(\pair{t}{u}, \abstr{x}v)) = 2 +
      \max(\mu(t),\mu(u)) + \mu(v) > \mu(u) + \mu(v) \geq \mu((u/x)v)$
    \item $\mu(\elimplus(\inl(t),\abstr{x}v,\abstr{y}w)) = 2 + \mu(t) +
      \max(\mu(v), \mu(w)) > \mu(t) + \mu(v) \geq \mu((t/x)v)$
    \item $\mu(\elimplus(\inr(t),\abstr{x}v,\abstr{y}w)) = 2 + \mu(t) +
      \max(\mu(v),\mu(w)) > \mu(t) + \mu(w) \geq \mu((t/y)w)$
    \item $\mu({a.\star} \plus b.\star) = 2 > 1 = \mu((a+b).\star)$
    \item $\mu((\lambda \abstr{x}t) \plus (\lambda \abstr{x}u)) = 1 + \max
      (1 + \mu(t), 1 + \mu(u)) = 2 + \max(\mu(t),\mu(u)) = \mu(\lambda
      \abstr{x}(t \plus u))$
    \item $\mu(\elimtens(t \plus u,\abstr{x y}v)) = 2 + \max(\mu(t),
      \mu(u)) + \mu(v) = 1 + \max(1 + \mu(t) + \mu(v), 1 + \mu(u) +
      \mu(v)) = \mu(\elimtens(t,\abstr{x y}v) \plus \elimtens(u,\abstr{x
      y}v))$
    \item $\mu(\topintro \plus \topintro) = 2 > 1 = \mu(\topintro)$
    \item $\mu(\pair{t}{u} \plus \pair{v}{w}) = 1 + \max (1 +
      \max(\mu(t),\mu(u)), 1 + \max(\mu(v),\mu(w)))$\\ $= 2 + \max
      (\mu(t),\mu(u),\mu(v),\mu(w)) = 1 + \max(1+ \max(\mu(t),\mu(v)), 1 +
      \max(\mu(u),\mu(w))) = \mu(\pair{t \plus v}{u \plus w})$
    \item $\mu(\elimplus(t \plus u,\abstr{x}v,\abstr{y}w)) = 2 +
      \max(\mu(t),\mu(u)) + \max(\mu(v), \mu(w)) = 1 + \max(1 + \mu(t) +
      \max(\mu(v), \mu(w)), 1 + \mu(u) + \max(\mu(v), \mu(w))) =
      \mu(\elimplus(t,\abstr{x}v,\abstr{y}w) \plus
      \elimplus(u,\abstr{x}v,\abstr{y}w))$
    \item $\mu(a \bullet b.\star) = 2 > 1 = \mu((a \times b).\star)$
    \item $\mu(a \bullet \lambda \abstr{x} t) = 2 + \mu(t) = \mu(\lambda
      \abstr{x} a \bullet t)$
    \item $\mu(\elimtens(a \bullet t,\abstr{x y}v)) = 2 + \mu(t) + \mu(v) =
      a \bullet \elimtens(t,\abstr{x y}v)$
    \item $\mu(a \bullet \topintro) = 2 > 1 = \mu(\topintro)$
    \item $\mu(a \bullet \pair{t}{u}) = 2 + \max(\mu(t),\mu(u)) = 1 +
      \max(1 + \mu(t),1 + \mu(u)) = \mu(\pair{a \bullet t}{a \bullet u})$

    \item
      $\mu(\elimplus(a \bullet t,\abstr{x}v,\abstr{y}w)) =
      2 + \mu(t) + \max(\mu(v),\mu(w)) =
      \mu(a \bullet \elimplus(t,\abstr{x}v,\abstr{y}w))$
      \qedhere
  \end{itemize} 
\end{proof}

\subsection{Elimination contexts}

The second part of the proof is a standard generalization of the
notion of head variable. In the $\lambda$-calculus, we can decompose a
term $t$ as a sequence of applications $t = u~v_1~\ldots~v_n$, with
terms $v_1, \dots, v_n$ and a term $u$, which is not an
application.  Then $u$ may either be a variable, in which case it is
the head variable of the term, or an abstraction.

Similarly, any proof in the ${\mathcal L}^{\mathcal
S}$-calculus can be decomposed into a sequence of elimination rules,
forming an elimination context, and a proof $u$ that is either a
variable, an introduction, a sum, or a product.

\begin{definition}[Elimination context]
  An elimination context is a proof with a single free variable, written
  $\_$, that is a proof in the language
  \[
    K = \_
    \mid \elimone(K,u)
    \mid K~u
    \mid \elimtens(K,\abstr{x y}v)
    \mid \elimzero(K)
    \mid \elimwith^1(K,\abstr{x}r)
    \mid \elimwith^2(K,\abstr{x}r)
    \mid \elimplus(K,\abstr{x}r,\abstr{y}s)
  \]
  where $u$ is a closed proof,
  $FV(v) = \{x,y\}$,
  $FV(r) \subseteq \{x\}$, and $FV(s) \subseteq \{y\}$.

\end{definition}

In the case of elimination contexts, Lemma~\ref{lem:msubst} can be
strengthened.

\begin{lemma}
  \label{strengtheningmsubst}
  $\mu(K\{t\}) = \mu(K) + \mu(t)$
\end{lemma}

\begin{proof}
  By induction on $K$

  \begin{itemize}
    \item If $K = \_$, then $\mu(K) = 0$ and $K\{t\} = t$.  We have
      $\mu(K\{t\}) = \mu(t) = \mu(K) + \mu(t)$.

    \item If $K = \elimone(K_1,u)$ then $K\{t\} = \elimone(K_1\{t\},u)$.
      We have, by
      induction hypothesis, $\mu(K\{t\}) = 1 + \mu(K_1\{t\}) + \mu(u)
      = 1 + \mu(K_1) + \mu(t) + \mu(u) = \mu(K) + \mu(t)$.

    \item If $K = K_1~u$ then $K\{t\} = K_1\{t\}~u$.  We have, by
      induction hypothesis, $\mu(K\{t\}) = 1 + \mu(K_1\{t\}) + \mu(u) = 1 +
      \mu(K_1) +   \mu(t) + \mu(u) = \mu(K) + \mu(t)$.

    \item 
      If $K = \elimtens(K_1,\abstr{xy}v)$, then $K\{t\} =
      \elimtens(K_1\{t\},\abstr{xy}v)$.  We have, by induction
      hypothesis, $\mu(K\{t\}) = 1 + \mu(K_1\{t\}) + \mu(v) = 1 +
      \mu(K_1) + \mu(t) + \mu(v) = \mu(K) + \mu(t)$.

    \item If $K = \elimzero(K_1)$, then $K\{t\} = \elimzero(K_1\{t\})$. We
      have, by induction hypothesis, $\mu(K\{t\}) = 1 + \mu(K_1\{t\})= 1 +
      \mu(K_1) + \mu (t) = \mu(K) + \mu(t)$.

    \item If $K = \elimwith^1(K_1,\abstr{x}r)$, then $K\{t\} =
      \elimwith^1(K_1\{t\},\abstr{x}r)$. We have, by induction hypothesis,
      $\mu(K\{t\}) = 1 + \mu(K_1\{t\}) + \mu(r) = 1 + \mu(K_1) + \mu (t) +
      \mu(r) = \mu(K) + \mu(t)$.

      The same holds if $K = \elimwith^2(K_1,\abstr{y}s)$.

    \item 
      If $K = \elimplus(K_1,\abstr{x}r,\abstr{y}s)$, then $K\{t\} =
      \elimplus(K_1\{t\},\abstr{x}r,\abstr{y}s)$.  We have, by induction
      hypothesis, $\mu(K\{t\}) = 1 + \mu(K_1\{t\}) + \max(\mu(r), \mu(s)) = 1 +
      \mu(K_1) + \mu(t) + \max(\mu(r), \mu(s)) = \mu(K) + \mu(t)$.
      \qedhere
  \end{itemize}
\end{proof}

Note that in Example~\ref{inferieurstrict}, $t = \elimzero(y)$ is not
an elimination context as $\_$ does not occur in it.  Note also that
the proof of Lemma~\ref{strengtheningmsubst} uses the fact that the
function $\mu$ is a mix between a size function and a depth function. The
corresponding lemma holds neither for the size function nor for the
depth function.

\begin{lemma}[Decomposition of a proof]
  \label{elim}
  If $t$ is an irreducible proof such that $x:C \vdash t:A$, then there
  exist an elimination context $K$, a proof $u$, and a proposition $B$,
  such that $\_:B \vdash K:A$, $x:C \vdash u:B$, $u$ is either the
  variable $x$, an introduction, a sum, or a product, and $t = K\{u\}$.
\end{lemma}

\begin{proof}
  By induction on the structure of $t$.

\begin{itemize}
    \item If $t$ is the variable $x$, an introduction, a sum, or a
     product, we take $K = \_$, $u = t$, and $B = A$.

    \item If $t = \elimone(t_1,t_2)$, then $t_1$ is not a closed proof as
      otherwise it would be a closed irreducible proof of $\one$, hence,
      by Theorem \ref{introductions}, it would be an introduction and $t$
      would not be irreducible. Thus, by the inversion property, $x:C
      \vdash t_1:\one$ and $\vdash t_2:A$.

      By induction hypothesis, there exist $K_1$, $u_1$ and $B_1$ such
      that $\_:B_1 \vdash K_1:\one$, $x:C \vdash u_1:B_1$, and $t_1 =
      K_1\{u_1\}$.  We take $u = u_1$, $K = \elimone(K_1,t_2)$, and $B =
      B_1$.  We have $\_:B \vdash K:A$, $x:C \vdash u:B$, and $K\{u\} =
      \elimone(K_1\{u_1\},t_2) = t$.
      
    \item If $t = t_1~t_2$,
          we apply the same method as for the case $t = \elimone(t_1,t_2)$.

    \item If $t = \elimtens(t_1,\abstr{yz}t_2)$, then $t_1$ is not a
      closed proof as otherwise it would be a closed irreducible proof of a
      multiplicative conjunction $\otimes$, hence,
      by Theorem \ref{introductions}, it would be an introduction, a sum, or a
      product, and $t$ would not be irreducible. Thus, by the inversion
      property, $x:C \vdash t_1:D_1 \otimes D_2$ and $y:D_1, z:D_2 \vdash
      t_2:A$.

      By induction hypothesis, there exist $\_:B_1 \vdash K_1:C \otimes
      D$, $x:C \vdash u_1:B_1$, and $t_1 = K_1\{u_1\}$.  We take $u =
      u_1$, $K = \elimtens(K_1,\abstr{y z}t_2)$, and $B = B_1$.  We have
      $\_:B \vdash K:A$, $x:C \vdash u:B$, and $K\{u\} =
      \elimtens(K_1\{u_1\},\abstr{y z}t_2) = t$.

    \item If $t = \elimzero(t_1)$, then,
      by Theorem \ref{introductions},
      $t_1$ is not a closed proof as there is
      no closed irreducible proof of $\zero$. Thus, by the inversion
      property, $x:C \vdash t_1:\zero$.

      By induction hypothesis, there exist $K_1$, $u_1$, and $B_1$ such
      that $\_:B_1 \vdash K_1:\zero$, $x:C \vdash u_1:B_1$, and $t_1 =
      K_1\{u_1\}$.  We take $u = u_1$, $K = \elimzero(K_1)$, and $B =
      B_1$.  We have $\_:B, \vdash K:A$, $x:C \vdash u:B$, and $K\{u\} =
      \elimzero(K_1\{u_1\}) = t$.

    \item If $t = \elimwith^1(t_1,\abstr{y}t_2)$ or
      $t = \elimwith^2(t_1,\abstr{y}t_2)$,
      we apply the same method as for the case $t = \elimone(t_1,t_2)$.

    \item If $t =  \elimplus(t_1,\abstr{y}t_2,\abstr{z}t_3)$, we apply the same
      method as for the case $t = \elimtens(t_1,\abstr{yz}t_2)$. \qedhere

\end{itemize}
\end{proof}

A final lemma shows that
we can always
decompose an elimination context $K$ different from~$\_$ into a
smaller elimination context $K_1$ and a last elimination rule $K_2$. 
This is similar to the fact that we can always decompose a non-empty list into
a smaller list and its last element.

\begin{lemma}[Decomposition of an elimination context]
  \label{horrible}
  If $K$ is an elimination context such that $\_:A \vdash K:B$
  and $K \neq \_$, then $K$ has the form $K_1\{K_2\}$
  where $K_1$ is an elimination context and
  $K_2$ is an elimination context formed with a single
  elimination rule, that is the elimination rule of the top symbol of $A$.
\end{lemma}

\begin{proof}
  As $K$ is not $\_$, it has the form $K = L_1\{L_2\}$.
  If $L_2 = \_$, we take $K_1 = \_$, $K_2 = L_1$ and, as the proof is
  well-formed, $K_2$ must be an elimination of the top symbol of $A$.
  Otherwise, by induction hypothesis, $L_2$ has the form $L_2 = K'_1\{K'_2\}$,
  and $K'_2$ is an elimination of the top symbol of $A$.
  Hence, $K = L_1\{K'_1\{K'_2\}\}$. We take $K_1 = L_1\{K'_1\}$, $K_2 = K'_2$.
\end{proof}

\subsection{Linearity}

We now have the tools to prove the linearity theorem expressing that if $A$
is a proposition, $B$ is a proposition of ${\mathcal V}$, $t$ is a proof
such that $x:A \vdash t:B$, and $u_1$ and $u_2$ are two closed proofs
of $A$, then
$$t\{u_1 \plus u_2\} \equiv t\{u_1\} \plus t\{u_2\}
\qquad\qquad\textrm{and}\qquad\qquad
t\{a \bullet u_1\} \equiv a \bullet t\{u_1\}$$

We proceed by induction on the measure $\mu(t)$ of the proof $t$, but
the case analysis is non-trivial. Indeed, when $t$ is an elimination,
for example when $t = t_1~t_2$, the variable $x$ must occur in $t_1$,
and we would like to apply the induction hypothesis to this proof.
But we cannot because $t_1$ is a proof of an implication, that is not
in $\mathcal{V}$.  This leads us to first decompose the proof $t$ into
a proof of the form $K\{t'\}$ where $K$ is an elimination context and
$t'$ is either the variable $x$, an introduction, a sum, or a product,
and analyse the different possibilities for $t'$. The cases where $t'$
is an introduction, a sum or a product are easy, but the case where it
is the variable $x$, that is where $t = K\{x\}$, is more complex.
Indeed, in this case, we need to prove
$$K\{u_1 \plus u_2\} \equiv K\{u_1\} \plus K\{u_2\}
\qquad\textrm{and}\qquad
K\{a \bullet u_1\} \equiv a \bullet K\{u_1\}$$
and this leads to a second case analysis where we consider the last elimination
rule of $K$ and how it interacts with $u_1$ and $u_2$.

For example, when $K = K_1 \{\elimwith^1(\_,\abstr{y}r)\}$, then $u_1$
and $u_2$ are closed proofs of an additive conjunction $\with $, thus they
reduce to two pairs $\pair{u_{11}}{u_{12}}$ and
$\pair{u_{21}}{u_{22}}$, and $K\{u_1 \plus u_2\}$ reduces to $K_1
\{r\}\{u_{11} \plus u_{21}\}$. So, we need to apply the induction
hypothesis to the irreducible form of $K_1 \{r\}$. To prove that this
proof is smaller than $t$, we need Lemma~\ref{lem:mured} (hence
Lemma~\ref{lem:msubst}) and Lemma~\ref{strengtheningmsubst}.

In fact, this general case has several exceptions: the cases of the
elimination of the multiplicative conjunction $\otimes$ and of the
additive disjunction $\oplus$ are simplified because, the sum commutes
with the elimination rules and not the introduction rules of these
connectives. The case of the elimination of the connective $\zero$ is
simplified because it is empty.  The case of the elimination of the
connective $\one$ is simplified because no substitution occurs in $r$
in this case. The case of the elimination of the implication is
simplified because this rule is just the modus ponens and not the
generalized elimination rule of this connective. Thus, the only remaining
cases are those of the elimination rules of the additive conjunction
$\with $.

\begin{theorem}[Linearity]
  \label{linearity}
  If $A$ is a proposition, $B$ is proposition of ${\mathcal V}$, $t$ is
  a proof such that $x:A \vdash t:B$, and $u_1$ and $u_2$ and two closed
  proofs of $A$, then
  $$t\{u_1 \plus u_2\} \equiv t\{u_1\} \plus t\{u_2\}
  \qquad\qquad\textrm{and}\qquad\qquad
  t\{a \bullet u_1\} \equiv a \bullet t\{u_1\}$$
\end{theorem}

\begin{proof}
  Without loss of generality, we can assume that $t$ is irreducible.
  We proceed by induction on $\mu(t)$. 

  Using Lemma \ref{elim}, the term $t$ can be decomposed as $K\{t'\}$
  where $t'$ is either the variable $x$, an introduction, a sum, or a
  product.

  \begin{itemize}
    \item 
      If $t'$ is an introduction, as $t$ is irreducible, $K = \_$ and
      $t'$ is a proof of $B \in
      {\mathcal V}$, $t'$ is either $a.\star$ or $\pair{t_1}{t_2}$. However,
      since $a.\star$ is not a proof in $x:A$, it is $\pair{t_1}{t_2}$.
      Using the induction hypothesis with $t_1$ and with $t_2$
      ($\mu(t_1) < \mu(t')$, $\mu(t_2) < \mu(t')$), 
      we get
      \begin{align*}
	t\{u_1 \plus u_2\}
	\equiv
	\pair{t_1\{u_1\} \plus t_1\{u_2\}}{t_2\{u_1\} \plus t_2\{u_2\}}
	\lla 
	t\{u_1\} \plus t\{u_2\}
      \end{align*}
      And
      \begin{align*}
	t\{a \bullet u_1\}
	\equiv \pair{a \bullet t_1\{u_1\}}{a \bullet t_2\{u_1\}}
	&
	\lla
	a \bullet t\{u_1\}
      \end{align*}

    \item If $t' = t_1 \plus t_2$, then using the induction hypothesis
      with $t_1$, $t_2$, and $K$ ($\mu(t_1) < \mu(t)$, $\mu(t_2) <
      \mu(t)$, and $\mu(K) < \mu(t)$) and Lemma~\ref{vecstructure} (1.,
      2., and 7.), we get
      \begin{align*}
	t\{u_1 \plus u_2\}
&	\equiv K\{(t_1\{u_1\} \plus t_1\{u_2\})
	\plus (t_2\{u_1\} \plus t_2\{u_2\})\}\\
	&\equiv 
	K\{(t_1\{u_1\} \plus t_2\{u_1\})
	\plus (t_1\{u_2\} \plus t_2\{u_2\})\}
	\equiv  
	t\{u_1\} \plus t\{u_2\}
      \end{align*}
      And 
      \begin{align*}
	t\{a \bullet u_1\}
&	\equiv K\{a \bullet t_1\{u_1\} \plus a \bullet t_2\{u_1\}\}
	\equiv K\{a \bullet (t_1\{u_1\} \plus t_2\{u_1\})\}
	\equiv 
	a \bullet t\{u_1\}
      \end{align*}

    \item
      If $t' = b \bullet t_1$, then using the induction hypothesis
      with $t_1$ and $K$ ($\mu(t_1) < \mu(t)$, $\mu(K) < \mu(t)$) and
      $K$ and Lemma~\ref{vecstructure} (7.~and 5.), we get
      \begin{align*}
	t\{u_1 \plus u_2\}
	\equiv K\{b \bullet (t_1 \{u_1\} \plus t_1\{u_2\})\}
	& \equiv K\{b \bullet t_1 \{u_1\} \plus b \bullet t_1\{u_2\}\}
	 \equiv 
	t\{u_1\} \plus t\{u_2\}
      \end{align*}
      And 
      \begin{align*}
	t\{a \bullet u_1\}
&	\equiv K\{b \bullet a \bullet t_1 \{u_1\}\}
\equiv K\{a \bullet b \bullet t_1 \{u_1\}\}
	\equiv 
	a \bullet t\{u_1\}
      \end{align*}

    \item If $t'$ is the variable $x$, we need to prove
      $$K\{u_1 \plus u_2\} \equiv K\{u_1\} \plus K\{u_2\}
      \qquad\textrm{and}\qquad
      K\{a \bullet u_1\} \equiv a \bullet K\{u_1\}$$
      By Lemma~\ref{horrible},
      $K$ has the form $K_1\{K_2\}$ and $K_2$ is an elimination of the top symbol of $A$.
      We consider the various cases for $K_2$. 
      \begin{itemize}
	\item If $K = K_1\{\elimone(\_,r)\}$, then $u_1$ and $u_2$ are closed
	  proofs of $\one$, thus $u_1 \lras b.\star$ and $u_2\lras c.\star$.
	  Using the induction hypothesis with the proof $K_1$
	  ($\mu(K_1) < \mu(K) = \mu(t)$) and Lemma~\ref{vecstructure} (8.~and 5.)
	  \begin{align*}
	    K\{u_1 \plus u_2\}
	    &\lras K_1 \{\elimone({b.\star} \plus c.\star,r)\}
	    \lras K_1 \{(b + c) \bullet r\}
	    \equiv (b + c) \bullet K_1 \{r\}\\
	    &\equiv b \bullet K_1 \{r\} \plus c \bullet K_1 \{r\}
	    \equiv K_1 \{b \bullet r\} \plus K_1 \{c \bullet r\}\\
	    & \llas K_1 \{\elimone(b.\star,r)\} \plus K_1 \{\elimone(c.\star,r)\}
	    \llas K\{u_1\} \plus K\{u_2\}
	  \end{align*}
	  And
	  \begin{align*}
	    K\{a \bullet u_1\}
	    &\lras K_1 \{\elimone(a \bullet b.\star,r)\}
	    \lras K_1 \{(a \times b) \bullet r\}
	    \equiv (a \times b) \bullet  K_1 \{r\}\\
	   & \equiv a \bullet b \bullet  K_1 \{r\}
	    \equiv a \bullet  K_1 \{b \bullet r\}
	     \llas a \bullet K_1 \{\elimone(b.\star,r)\}
	    \llas a \bullet K\{u_1\}
	  \end{align*}

	\item
	  If $K = K_1 \{\_~s\}$, then $u_1$ and $u_2$ are closed 
	  proofs of an implication, thus 
	  $u_1 \lras \lambda \abstr{y} u'_1$
	  and
	  $u_2 \lras \lambda \abstr{y} u'_2$.
	  Using the induction hypothesis with the proof $K_1 $
	  ($\mu(K_1 ) < \mu(K) = \mu(t)$), we get
	  \begin{align*}
	    K\{u_1 \plus u_2\}
	    &\lras K_1 \{(\lambda \abstr{y} u'_1 \plus \lambda \abstr{y} u'_2)~s\}
	    \lras K_1 \{u'_1\{s\} \plus u'_2\{s\}\}\\
	    &\equiv K_1 \{u'_1\{s\}\} \plus K_1 \{u'_2\{s\}\}
	     \llas K_1 \{(\lambda \abstr{y} u'_1)~s\} \plus K_1 \{(\lambda \abstr{y} u'_2)~s\}\\
	   & \llas K\{u_1\} \plus K\{u_2\}
	  \end{align*}
	  And
	  \begin{align*}
	    K\{a \bullet u_1\}
	    &\lras K_1 \{(a \bullet \lambda \abstr{y} u'_1)~s\}
	    \lras K_1 \{a \bullet u'_1\{s\}\}\\
	   & \equiv a \bullet  K_1 \{u'_1\{s\}\}
	     \lla
	    a \bullet K_1 \{(\lambda \abstr{y} u'_1)~s\}
	    \llas a \bullet K\{u_1\}
	  \end{align*}

	\item
	  If $K = K_1 \{\elimtens(\_,\abstr{y z}r)\}$, then, using
	  the induction hypothesis with the proof $K_1 $ 
	  ($\mu(K_1 ) < \mu(K) = \mu(t)$), we get
	  \begin{align*}
	    K\{u_1 \plus u_2\}
	    \lra K_1 \{\elimtens(u_1,\abstr{y z}r)
	    \plus \elimtens(u_2,\abstr{y z}r)\}
	    &\equiv
            K\{u_1\} \plus K\{u_2\}
	  \end{align*}
	  And 
	  \begin{align*}
	    K\{a \bullet u_1\}
	    \lra K_1 \{a \bullet \elimtens(u_1,\abstr{y z}r)\}
	    &\equiv
            a \bullet K\{u_1\}
	  \end{align*}

	\item The case $K = K_1 \{\elimzero(\_)\}$ is not possible as $u_1$ would
	  be a closed proof of $\zero$ and there is no such proof.

	\item
	  If $K = K_1 \{\elimwith^1(\_,\abstr{y}r)\}$, then $u_1$ and $u_2$
	  are closed proofs of an additive conjunction $\with $, thus $u_1 \lras
	  \pair{u_{11}}{u_{12}}$ and $u_2 \lras \pair{u_{21}}{u_{22}}$.

	  Let $r'$ be the irreducible form of $K_1 \{r\}$.
	  Using the induction hypothesis with the proof $r'$
	  (because, with Lemmas~\ref{lem:mured} and~\ref{strengtheningmsubst}, we have
	    $\mu(r') \leq \mu(K_1 \{r\}) = \mu(K_1 ) + \mu(r) < \mu(K_1 ) + \mu(r) + 1
	  = \mu(K) = \mu(t)$)
	  \begin{align*}
	    K\{u_1 \plus u_2\}
	    &\lras K_1 \{\elimwith^1(\pair{u_{11}}{u_{12}} \plus \pair{u_{21}}{u_{22}}, \abstr{y}r)\}
	    \lras K_1 \{r\{u_{11} \plus u_{21}\}\}\\
	   & \lras r'\{u_{11} \plus u_{21}\}
	    \equiv r'\{u_{11}\} \plus r'\{u_{21}\}
	    \llas K_1 \{r\{u_{11}\}\} \plus K_1 \{r\{u_{21}\}\}\\
	   & \llas K_1 \{\elimwith^1(\pair{u_{11}}{u_{12}},\abstr{y}r)\}\plus K_1 \{\elimwith^1(\pair{u_{21}}{u_{22}},\abstr{y}r)\}
     \\
	    &\llas K\{u_1\} \plus K\{u_2\} 
	  \end{align*}
	  And
	  \begin{align*}
	    K\{a \bullet u_1\}
	    &\lras K_1 \{\elimwith^1(a \bullet \pair{u_{11}}{u_{12}},\abstr{y}r)\}
	    \lra^* K_1 \{r\{a \bullet u_{11}\}\}
	    \lras r'\{a \bullet u_{11}\}\\
	   & \equiv a \bullet r'\{u_{11}\}
	    \llas a \bullet K_1 \{r\{u_{11}\}\}
	    \lla a \bullet K_1 \{\elimwith^1(\pair{u_{11}}{u_{12}},\abstr{y}r)\}\\
	   & \llas a \bullet K\{u_1\}
	  \end{align*}

	\item If $K = K_1 \{\elimwith^2(\_,\abstr{y}r)\}$, the proof is similar.

	\item
	  If $K = K_1 \{\elimplus(\_,\abstr{y}r, \abstr{z}s)\}$, then, using
	  the induction hypothesis with the proof $K_1 $ 
	  ($\mu(K_1 ) < \mu(K) = \mu(t)$), we get
	  \begin{align*}
	    K\{u_1 \plus u_2\}
	    \lra K_1 \{\elimplus( u_1,\abstr{y}r, \abstr{z}s) \plus \elimplus( u_2,\abstr{y}r, \abstr{z}s)\}
	    &\equiv
            K\{u_1\} \plus K\{u_2\}
	  \end{align*}
	  And 
	  \begin{align*}
	    K\{a \bullet u_1\}
	    \lra K_1 \{a \bullet \elimplus(u_1,\abstr{y}r, \abstr{z}s)\}
	    &
	    \equiv
	    a \bullet K\{u_1\}
	    \tag*{\qedhere}
	  \end{align*}
      \end{itemize}
  \end{itemize}
\end{proof}

We can now generalize the linearity result, as explained in Section~\ref{sec:obseq}, by using the observational equivalence $\sim$ (cf.~Definition~\ref{def:obseq}).

\begin{corollary}\label{corollary0}
  If $A$ and $B$ are any propositions,
  $t$ a proof such that $x:A \vdash t:B$ and
  $u_1$ and $u_2$ two closed proofs of $A$, then  
  $$t\{u_1 \plus u_2\} \sim t\{u_1\} \plus t\{u_2\}
  \qquad\qquad\textrm{and}\qquad\qquad
  t\{a \bullet u_1\} \sim a \bullet t\{u_1\}$$
\end{corollary}

\begin{proof}
  Let $C \in {\mathcal V}$ and $c$ be a proof such that
  $\_:B \vdash c:C$. Then applying Theorem \ref{linearity} to
  the proof $c\{t\}$ we get 
  $$c\{t\{u_1 \plus u_2\}\} \equiv c\{t\{u_1\}\} \plus c\{t\{u_2\}\}
  \qquad\qquad\textrm{and}\qquad\qquad
  c\{t\{a \bullet u_1\}\} \equiv a \bullet c\{t\{u_1\}\}$$
  and applying it again to the proof $c$ we get
  $$c\{t\{u_1\} \plus t\{u_2\}\} \equiv c\{t\{u_1\}\} \plus c\{t\{u_2\}\}
  \qquad\qquad\textrm{and}\qquad\qquad
  c\{a \bullet t\{u_1\}\} \equiv a \bullet c\{t\{u_1\}\}$$
  Thus
  $$c\{t\{u_1 \plus u_2\}\} \equiv c\{t\{u_1\} \plus t\{u_2\}\}
  \qquad\qquad\textrm{and}\qquad\qquad
  c\{t\{a \bullet u_1\}\} \equiv c\{a \bullet t\{u_1\}\}$$
  that is 
  \begin{align*}
    t\{u_1 \plus u_2\} \sim t\{u_1\} \plus t\{u_2\}
    &\qquad\qquad\textrm{and}\qquad\qquad
    t\{a \bullet u_1\} \sim a \bullet t\{u_1\}
    \tag*{\qedhere}
  \end{align*}
\end{proof}

The main result, as announced in Section~\ref{sec:obseq}, showing that proofs of $A\multimap B$ are linear functions, is a direct consequence of Theorem~\ref{linearity} and Corollary~\ref{corollary0}.

\begin{corollary}
  \label{corollary1}
  Let $A$ and $B$ be propositions. 
  Let $t$ be a closed proof of $A \multimap B$ and
  $u_1$ and $u_2$ be closed proofs of $A$.

  Then, if $B \in {\mathcal V}$, we have 
  $$t~(u_1 \plus u_2) \equiv (t~u_1) \plus (t~u_2)
  \qquad\qquad\textrm{and}\qquad\qquad
  t~(a\bullet u_1) \equiv a\bullet (t~u_1)$$
  and in the general case, we have 
  $$t~(u_1 \plus u_2) \sim (t~u_1) \plus (t~u_2)
  \qquad\qquad\textrm{and}\qquad\qquad
  t~(a\bullet u_1) \sim a\bullet (t~u_1)$$
\end{corollary}

\begin{proof}
  As $t$ is a closed proof of $A \multimap B$, using
  Theorem~\ref{introductions}, it reduces to an irreducible proof of the
  form $\lambda \abstr{x} t'$.  Let $u'_1$ be the irreducible form of
  $u_1$, and $u'_2$ that of $u_2$.

  If $B \in {\mathcal V}$, using Theorem~\ref{linearity}, we have
\begin{align*}
      t~(u_1 \plus u_2) \lras t'\{u'_1 \plus u'_2\} 
      &\equiv t'\{u'_1\} \plus t'\{u'_2\} \llas (t~u_1) \plus (t~u_2)
      \\
      t~(a\bullet u_1) \lras t'\{a \bullet u'_1\} 
      &\equiv a \bullet t'\{u'_1\} \llas a\bullet (t~u_1)
\end{align*}
  In the general case, using Corollary \ref{corollary0}, we have
\begin{align*}
    t~(u_1 \plus u_2) \lras t'\{u'_1 \plus u'_2\}
    &\sim t'\{u'_1\} \plus t'\{u'_2\} \llas (t~u_1) \plus (t~u_2)
    \\
      t~(a\bullet u_1) \lras t'\{a \bullet u'_1\}
      &\sim a \bullet t'\{u'_1\} \llas a\bullet (t~u_1)
      \tag*{\qedhere}
\end{align*}
\end{proof}

Finally, the next corollary is the converse of Theorem \ref{matrices}.

\begin{corollary}
  \label{corollary2}
  Let $A, B \in {\mathcal V}$, such that $d(A) = m$ and $d(B) = n$, 
  and  $t$ be a closed proof of $A \multimap B$.
  Then the function $F$ from ${\mathcal S}^m$ to ${\mathcal S}^n$,
  defined as
  $F({\bf u}) = \underline{t~\overline{\bf u}^A}$ is linear.
\end{corollary}

\begin{proof}
  Using Corollary~\ref{corollary1} and Lemmas~\ref{parallelsum}
  and~\ref{parallelprod}, we have
\begin{align*}
      F({\bf u} + {\bf v}) 
      = \underline{t~\overline{\bf u + \bf v}^A}
      = \underline{t~(\overline{\bf u}^A \plus \overline{\bf v}^A)} 
      &= \underline{t~\overline{\bf u}^A \plus t~\overline{\bf v}^A}
      = \underline{t~\overline{\bf u}^A} + \underline{t~\overline{\bf v}^A}
      = F({\bf u}) + F({\bf v})
      \\
      F(a {\bf u}) 
      = \underline{t~\overline{a \bf u}^A}
      = \underline{t~(a \bullet \overline{\bf u}^A)} 
      &= \underline{a\bullet t~\overline{\bf u}^A}
      = a \underline{t~\overline{\bf u}^A}
      = a F({\bf u})
      \tag*{\qedhere}
\end{align*}
\end{proof}

\begin{remark} The Theorem \ref{linearity} and its corollaries hold for
  linear proofs, but not for non-linear ones. The linearity is used,
  in an essential way, in two places. First, in the first case of the
  proof of Theorem \ref{linearity} when we remark that $a.\star$ is
  not a proof in the context $x:A$. Indeed, if $t$ could be
  $a.\star$ then linearity would be violated as $1.\star\{4.\star
  \plus 5.\star\} = 1.\star$ and $1.\star\{4.\star\} \plus
  1.\star\{5.\star\} \equiv 2.\star$.  Then, in the proof of Lemma
  \ref{elim}, we remark that when $t_1~t_2$ is a proof in the context
  $x:A$ then $x$ must occur in $t_1$, and hence it does not occur in
  $t_2$, that is therefore closed. This way, the proof $t$ eventually
  has the form $K\{u\}$ and this would not be the case if $x$ could
  occur in $t_2$ as well.
\end{remark}

\subsection{No-cloning}

In the proof language of propositional intuitionistic logic extended with
interstitial rules and scalars, but with structural rules, the cloning function
from ${\mathcal S}^2$ to ${\mathcal S}^4$, mapping $\left(\begin{smallmatrix}
a\\b\end{smallmatrix}\right)$ to $\left(\begin{smallmatrix}
a^2\\ab\\ab\\b^2\end{smallmatrix}\right)$, which is the tensor product of the
vector $\left(\begin{smallmatrix} a\\b\end{smallmatrix}\right)$ with itself,
can be expressed \cite{odot}. But the proof given there is not the proof of a
proposition of the ${\mathcal L}^{\mathcal S}$-calculus.

Moreover, by Corollary~\ref{corollary2}, no proof
of $(\one \with  \one) \multimap ((\one \with  \one) \with  (\one \with  \one))$,
in the ${\mathcal L}^{\mathcal S}$-calculus, can express this
function, because it is not linear.

\section{The \texorpdfstring{${\mathcal L} \odot^{\mathbb C}$}{L-sup-S}-calculus and its application
to quantum computing}
\label{seclss}

There are two issues in the design of a quantum programming
language. The first, addressed in this paper, is to take into account
the linearity of the unitary operators and, for instance, avoid
cloning.  The second, addressed in \cite{odot}, is to express the
information-erasure, non-reversibility, and non-determinism of the
measurement.

\subsection{The connective \texorpdfstring{$\odot$}{sup}}

In \cite{odot}, we have introduced, besides interstitial rules and
scalars, a new connective $\odot$ (read ``sup'' for ``superposition'')
and given the type $\Q_1 = \one \odot \one$ to quantum bits, that is
superpositions of bits.

As to express the superposition $\alpha . \ket 0 + \beta .  \ket 1$,
we need both $\ket 0$ and $\ket 1$, the connective $\odot$ has the
introduction rule of the conjunction.  And as the measurement in the
basis $\ket 0$, $\ket 1$ yields either $\ket 0$ or $\ket 1$, the
connective $\odot$ has the elimination rule of the disjunction.  But,
to express quantum algorithms, we also need to transform qubits, using
unitary operators and, to express these operators, we require other
elimination rules for the connective $\odot$, similar to those of the
conjunction.

Thus, the connective $\odot$ has an introduction rule $\odot$-i similar
to that of the conjunction, one elimination rule $\odot$-e similar to
that of the disjunction, used to express the information-erasing,
non-reversible, and non-deterministic quantum measurement operators,
and two elimination rules $\odot$-e1 and $\odot$-e2 similar to those
of the conjunction, used to express the information-preserving,
reversible, and deterministic unitary operators.

The $\odot^{\mathbb C}$-calculus can express
quantum algorithms, including those using measurement, but as the use
of variables is not restricted, it can also express non-linear
functions, such as cloning operators.

We can thus mix the two ideas and introduce the 
${\mathcal L} \odot^{\mathcal S}$-calculus that is an extension of the
${\mathcal L}^{\mathcal S}$-calculus with a $\odot$ connective and also
a linear restriction of the 
$\odot^{\mathcal S}$-calculus. 
The ${\mathcal L}\odot^{\mathcal S}$-calculus is obtained by adding
the symbols $\super{.}{.}$, $\elimsup^1$, $\elimsup^2$,
$\elimsup$,
the deduction rules of Figure~\ref{odot}, and the reduction rules of
Figure~\ref{odotred}, to the ${\mathcal L}^{\mathcal S}$-calculus.

\begin{figure}[t]
  $$\irule{\Gamma \vdash t:A & \Gamma \vdash u:B}
  {\Gamma \vdash \super{t}{u}:A \odot B}
  {\mbox{$\odot$-i}}$$
  $$\irule{\Gamma \vdash t:A \odot B & \Delta, x:A \vdash u:C}
  {\Gamma, \Delta \vdash \elimsup^1(t,\abstr{x}u):C}
  {\mbox{$\odot$-e1}}$$
  $$\irule{\Gamma \vdash t:A \odot B & \Delta, x:B \vdash u:C}
  {\Gamma, \Delta \vdash \elimsup^2(t,\abstr{x}u):C}
  {\mbox{$\odot$-e2}}$$
  $$\irule{\Gamma \vdash t:A \odot B & \Delta, x:A \vdash u:C & \Delta, y:B \vdash v:C}
  {\Gamma, \Delta \vdash \elimsup(t,\abstr{x}u,\abstr{y}v):C}
  {\mbox{$\odot$-e}}$$
  \caption{The deduction rules of the ${\mathcal L}\odot^{\mathcal S}$-calculus.\label{odot}}
\end{figure}

\begin{figure}[t]
  \[
    \begin{array}{r@{\,}l}
      \elimsup^1(\super{t}{u}, \abstr{x}v) & \longrightarrow  (t/x)v \\
      \elimsup^2(\super{t}{u}, \abstr{x}v) & \longrightarrow  (u/x)v \\
      \elimsup(\super{t}{u},\abstr{x}v,\abstr{y}w) & \longrightarrow  (t/x)v \\
      \elimsup(\super{t}{u},\abstr{x}v,\abstr{y}w) & \longrightarrow  (u/y)w \\
      \\
      \super{t}{u} \plus \super{v}{w} & \longrightarrow  \super{t \plus v}{u \plus w}\\
      a \bullet \super{t}{u} &\longrightarrow  \super{a \bullet t}{a \bullet u}
    \end{array}
  \]
  \caption{The reduction rules of the ${\mathcal L}\odot^{\mathcal S}$-calculus. \label{odotred}}
\end{figure}

We use the symbols $\super{.}{.}$, $\elimsup^1$ and $\elimsup^2$, to
express vectors and matrices, just like in
Section~\ref{secvectorsmatrices}, except that the conjunction $\with$
is replaced with the connective $\odot$.

As the symbol $\elimsup$ enables to express measurement operators that
are not linear, we cannot expect to have an analogue of Corollary
\ref{corollary1} for the full ${\mathcal L} \odot^{\mathcal
  S}$-calculus---more generally, we cannot expect a calculus to both
enjoy such a linearity property and express the measurement operators.
Thus, the best we can expect is a linearity property for the
restriction of the ${\mathcal L} \odot^{\mathcal S}$-calculus,
excluding the $\elimsup$ symbol.  But, this result is a trivial
consequence of Corollary \ref{corollary1}, as if the $\odot$-e rule is
excluded, the connective $\odot$ is just a copy of the additive
conjunction $\with $. So, we shall not give a full proof of this theorem.

In the same way, the subject reduction proof of the ${\mathcal L}
\odot^{\mathcal S}$-calculus is similar to the proof of Theorem
\ref{th:subjectreduction} and the strong termination proof of the
${\mathcal L} \odot^{\mathcal S}$-calculus is similar to the proof of
Corollary \ref{termination}, with a few extra lemmas proving the
adequacy of the introduction and elimination symbols of the $\odot$
connective, similar to those of the strong termination proof of
\cite{odot}, so we shall not repeat this proof. In contrast, the
confluence property is lost, because the reduction rules of the
${\mathcal L} \odot^{\mathcal S}$-calculus are non-deterministic.

Thus, we shall focus in this section on an informal discussion on how
the ${\mathcal L} \odot^{\mathbb C}$-calculus can be used as a
quantum programming language.

\subsection{The \texorpdfstring{${\mathcal L}\odot^{\mathbb  C}$}{L-sup-C}-calculus
  as a quantum programming language}

We first express the vectors and matrices like in
Section~\ref{secvectorsmatrices}, except that we use the connective
$\odot$ instead of $\with $.  In particular the $n$-qubits, for $n
\geq 1$, are expressed, in the basis $\ket{0 \ldots 00}, \ket{0 \ldots
  01}, \ldots \ket{1 \ldots 11}$, as elements of ${\mathbb C}^{2^n}$,
that is as proofs of the vector proposition $\Q_n$ defined by
induction on $n$ as follows: $\Q_0 = \one$ and $\Q_{n+1} = \Q_n \odot
\Q_n$. For example, the proposition $\Q_2$ is $(\one \odot \one) \odot
(\one \odot \one)$, and the proof
$\super{\super{a.\star}{b.\star}}{\super{c.\star}{d.\star}}$
represents the vector $a . \ket{00} + b . \ket{01} + c . \ket{10} + d
. \ket{11}$. For instance, the vector $\frac{1}{\sqrt{2}} \ket{00} +
\frac{1}{\sqrt{2}}\ket{11}$ is represented by the proof
$\super{\super{\frac{1}{\sqrt{2}}.\star}{0.\star}}{\super{0.\star}{\frac{1}{\sqrt{2}}.\star}}$.

It has been shown \cite{odot} that the $\odot^{\mathbb C}$-calculus
with a reduction strategy restricting the reduction of
$\elimsup(\super{t}{u},\abstr{x}v,\abstr{y}w)$ to the cases where $t$
and $u$ are closed irreducible proofs, can be used to express quantum
algorithms.  We now show that the same holds for the ${\mathcal
  L}\odot^{\mathbb C}$-calculus.

As we have already seen how to express linear maps in the 
${\mathcal L}\odot^{\mathbb  C}$-calculus, we now turn to the expression of
the measurement operators.

\begin{definition}[Norm of a vector]
If $t$ is a closed irreducible proof of $\Q_n$, we define the
square of the norm $\|t\|^2$ of $t$ by induction on $n$.
\begin{itemize}
  \item If $n = 0$, then $t = a.\star$ and we take $\|t\|^2 = |a|^2$.
  \item  If $n = n'+1$, then $t = \super{u_1}{u_2}$ and we take
    $\|t\|^2 = \|u_1\|^2 + \|u_2\|^2$.
\end{itemize}
\end{definition}

We take the convention that any closed irreducible proof $u$ of
$\Q_{n}$, expressing a non-zero vector $\underline{u} \in
{\mathbb C}^{2^n}$, is an alternative expression of the $n$-qubit
$\frac{\underline{u}}{\|\underline{u}\|}$.  For example, the qubit
$\frac{1}{\sqrt{2}} . \ket{0} + \frac{1}{\sqrt{2}} . \ket{1}$ is
expressed as the proof
$\super{\frac{1}{\sqrt{2}}.\star}{\frac{1}{\sqrt{2}}.\star}$, but also
as the proof $\super{1.\star}{1.\star}$.

\begin{definition}[Probabilistic reduction]
  Probabilities are assigned to the
non-deterministic reductions of closed proofs of the form
$\elimsup(u,\abstr{x}v,\abstr{y}w)$ as follows.
A proof of the form
$\elimsup(\super{u_1}{u_2},\abstr{x}v,\abstr{y}w)$ where $u_1$ and
$u_2$ are closed irreducible proofs of $\Q_n$ reduces to
$(u_1/x)v$ with probability
$\frac{\|u_1\|^2}{\|u_1\|^2 +
\|u_2\|^2}$ and to 
$(u_2/y)w$ with probability
$\frac{\|u_2\|^2}{\|u_1\|^2 + \|u_2\|^2}$,
when $\|u_1\|^2$ and $\|u_2\|^2$ are not both
$0$.
When $\|u_1\|^2 = \|u_2\|^2 = 0$, or $u_1$ and $u_2$ are proofs of
propositions of a different form, we assign any probability, for
example $\tfrac{1}{2}$, to both reductions.
\end{definition}

\begin{definition}[Measurement operator]
If $n$ is a non-zero natural number, we define the measurement
operator $\pi_n$, measuring the first qubit of an $n$-qubit, as the
proof
$$\pi_n = \lambda \abstr{x} \elimsup(x, \abstr{y}
\super{y}{0_{\Q_{n-1}}},\abstr{z} \super{0_{\Q_{n-1}}}{z})$$
of the proposition $\Q_{n} \multimap \Q_{n}$,
where the proof $0_{\Q_{n-1}}$
  is given in Definition \ref{def:zerovector}.
\end{definition}

\begin{remark}
If $t$ is a closed irreducible proof of $\Q_{n}$ of
the form $\super{u_1}{u_2}$, such that $\|t\|^2 = \|u_1\|^2 +
\|u_2\|^2 \neq 0$, expressing the state of an $n$-qubit, then the
proof $\pi_n~t$ of the proposition $\Q_{n}$ reduces, with
probabilities $\tfrac{\|u_1\|^2}{\|u_1\|^2 + \|u_2\|^2}$ and
$\tfrac{\|u_2\|^2}{\|u_1\|^2 + \|u_2\|^2}$ to
$\super{u_1}{0_{\Q_{n-1}}}$ and to $\super{0_{\Q_{n-1}}}{u_2}$, that are the
states of the $n$-qubit, after the partial measure of the first qubit.
\end{remark}

The measurement operator $\pi_n$ returns the 
states of the $n$-qubit, after the partial measure of the first qubit.
We now show that it is also possible to return the classical result of
the measure, that is a Boolean.

\begin{definition}[Booleans]
Let ${\bf 0}$ and ${\bf 1}$ be the closed proofs of the 
proposition $\B = \one \oplus \one$
${\bf 0}=\inl(1.\star)$ and ${\bf
  1}=\inr(1.\star)$.
\end{definition}

As we do not have a weakening rule, we cannot define 
this measurement operator as
$$\lambda \abstr{x} \elimsup(x, \abstr{y}{\bf 0},\abstr{z}{\bf 1})$$
that maps all proofs of the form
$\super{u_1}{u_2}$ to ${\bf 0}$ or ${\bf 1}$ with the same probabilities as
above.
So, we continue to consider proofs modulo renormalization, that
is that any proof of the form $a \bullet {\bf 0}$ also
represents the Boolean ${\bf 0}$ and any proof of the form $b \bullet
{\bf 1}$ also represents the Boolean ${\bf 1}$.

\begin{definition}[Classical measurement operator]
If $n$ is a non-zero natural number, we define the measurement
operator $\pi'_n$, 
as the proof
$$
\pi'_n = \lambda x.\elimsup(x,\abstr{y}{\delta^{\Q_{n-1}}(y,{\bf 0})},\abstr{z}{\delta^{\Q_{n-1}}(z,{\bf 1})})
$$
where $\delta^{\Q_{n}}$ is defined as 
$$
\delta^{\Q_{n}}(x,{\bf b}) =
\left\{
  \begin{array}{ll}
    \elimone(x,{\bf b}) & \textrm{if }n=0\\
    \elimsup(x,\abstr{y}{\delta^{\Q_{n-1}}(y,{\bf b})},\abstr{z}{\delta^{\Q_{n-1}}(z,{\bf b})}) & \textrm{if }n>0
  \end{array}
\right.
$$
\end{definition}

\begin{remark}
If $t$ is a closed irreducible proof of $\Q_{n}$ of the form
$\super{u_1}{u_2}$, such that $\|t\|^2 = \|u_1\|^2 + \|u_2\|^2 \neq
0$, expressing the state of an $n$-qubit, then the proof $\pi'_n~t$ of
the proposition $\B$ reduces, with the same probabilities as above, 
to $a \bullet {\bf 0}$ or $b \bullet {\bf 1}$.  The scalars $a$ and
$b$ may vary due to the probabilistic nature of the operator
$\elimsup$, but they are $0$ only with probability $0$.
\end{remark}

\begin{example}
The operator
$$\pi'_1 = \lambda \abstr{x} \elimsup(x,\abstr{y}{\elimone(y,\bf 0)},\abstr{z}{\elimone(z,
\bf 1)})
$$
applied to the proof $\super{a.\star}{b.\star}$ yields
$a\bullet\bf 0$ or $b\bullet\bf 1$ with probability $\frac{|a|^2}{|a|^2+|b|^2}$ or $\frac{|b|^2}{|a|^2+|b|^2}$ respectively.

The operator
$$\pi'_2 = \lambda \abstr{x} \elimsup(x,\abstr{y}{\delta^{\Q_1}(y,\bf 0)},\abstr{z}{\delta^{\Q_1}(z, \bf 1)})$$
applied to $\super{\super{a.\star}{b.\star}}{\super{c.\star}{d.\star}}$
reduces to
$\delta^{\Q_1}([a.\star,b.\star],\bf 0)$
or to 
$\delta^{\Q_1}([c.\star,d.\star],\bf 1)$
with probabilities $\frac{|a|^2 + |b|^2}{|a|^2 + |b|^2 + |c|^2 + |d|^2}$
and $\frac{|c|^2 + |d|^2}{|a|^2 + |b|^2 + |c|^2 + |d|^2}$.

Then the first proof   
$$\delta^{\Q_1}([a.\star,b.\star],\bf 0)$$
always
reduces to ${\bf 0}$ modulo some scalar multiplication, precisely to
$a \bullet {\bf 0}$ with probability 
$\frac{|a|^2}{|a|^2 + |b|^2}$
and
$b \bullet {\bf 0}$ with probability 
$\frac{|b|^2}{|a|^2 + |b|^2}$.
In the same way, the second always
reduces to {\bf 1} modulo some scalar multiplication. 
\end{example}

\subsection{Deutsch's algorithm}

We have given in \cite{odot}, a proof that expresses Deutsch's algorithm.
We update it here to the linear case.

As above, let ${\bf 0} = \inl(1.\star)$ and ${\bf 1} = \inr(1.\star)$
be closed irreducible proofs of $\B = \one \oplus \one$.

For each proposition $A$, and pair of closed terms, $u$ and $v$, of
type $A$, we have a test operator, that is a proof of $\B \multimap A$
$$\test_{u,v} = \lambda \abstr{x} 
\elimplus(x,\abstr{w_1}\elimone(w_1,u),\abstr{w_2}\elimone(w_2,v))$$
Then 
$\test_{u,v}~{\bf 0} \longrightarrow 1 \bullet u$ and $\test_{u,v}~{\bf
1} \longrightarrow 1 \bullet v$.
Deutsch's algorithm is the proof of
$(\B \multimap \B) \multimap \B$
$$\mbox{\it Deutsch}
=
\lambda \abstr{f}\pi'_2 ((\overline{H\otimes I})~(U~f~\overline{\ket{+-}}))$$
where $\overline{H \otimes I}$ is the proof of $\Q_2 \multimap Q_2$ corresponding to
the matrix
$$
\frac 1{\sqrt 2}\left( \begin{smallmatrix}
    1 & 0 & 1 & 0 \\
    0 & 1 & 0 & 1 \\
    1 & 0 & -1 & 0 \\
    0 & 1 & 0 & -1 
\end{smallmatrix}\right)
$$
as in Theorem \ref{matrices}, 
except that the
conjunction $\with $ is replaced with the connective $\odot$.
$U$ is the proof of $(\B \multimap \B) \multimap \Q_2 \multimap Q_2$  
\begin{align*}
  U = \lambda \abstr{f} \lambda \abstr{t}
  &\elimsup^1
  (
    t,\abstr{x}
    (
      \elimsup^1(x,\abstr{z_0}M_0~z_0)
      \plus
      \elimsup^2(x,\abstr{z_1}M_1~z_1)
    )
  )
  \\
  \plus~&
  \elimsup^2
  (
    t,
    \abstr{y}
    (
      \elimsup^1(y,\abstr{z_2}~M_2~z_2)
      \plus
      \elimsup^2(y,\abstr{z_3}~M_3~z_3)
    )
  )
\end{align*}
where $M_0$, $M_1$, $M_2$, and $M_3$ are the 
proofs of $\one \multimap \Q_2$
$$M_0 = \lambda \abstr{s} \elimone(s,
\test_{\super{\super{1.\star}{0.\star}}{\super{0.\star}{0.\star}} ,
  \super{\super{0.\star}{1.\star}}{\super{0.\star}{0.\star}}} ~(f~{\bf
  0}))$$
$$M_1 = \lambda \abstr{s} \elimone(s,
\test_{\super{\super{0.\star}{1.\star}}{\super{0.\star}{0.\star}},
  \super{\super{1.\star}{0.\star}}{\super{0.\star}{0.\star}}} ~(f~{\bf
  0}))$$
$$M_2 = \lambda \abstr{s} \elimone(s,
\test_{\super{\super{0.\star}{0.\star}}{\super{1.\star}{0.\star}},
  \super{\super{0.\star}{0.\star}}{\super{0.\star}{1.\star}}} ~(f~{\bf
  1}))$$
$$M_3 = \lambda \abstr{s} \elimone(s,
\test_{\super{\super{0.\star}{0.\star}}{\super{0.\star}{1.\star}},
  \super{\super{0.\star}{0.\star}}{\super{1.\star}{0.\star}}} ~(f~{\bf
  1}))$$
and $\overline{\ket{+-}}$ is the proof of $\Q_2$
$$\overline{\ket{+-}} =
\super{\super{\frac{1}{2}.\star}{\frac{-1}{2}.\star}}
{\super{\frac{1}{2}.\star}{\frac{-1}{2}.\star}}$$

Let $f$ be a proof of $\B \multimap \B$.
If $f$ is a constant function, we have $\mbox{\it
Deutsch}\ f \lras a \bullet {\bf 0}$, for some scalar $a$, 
while if $f$ is not constant, $\mbox{\it Deutsch}\
f \lras a \bullet {\bf 1}$ for some scalar $a$.

\subsection{Towards unitarity}
For future work, we may want to restrict the logic further so that functions
are not only linear, but also unitary.  Unitarity, the property that ensures
that the norm and orthogonality of vectors is preserved, is a requirement for
quantum gates.  In the current version, we can argue that we let these
unitarity constraints as properties of the program that must be proved for each
program, rather than enforced by the type system.

Some methods to enforce unitarity in quantum controlled lambda calculus has
been given in
\cite{AltenkirchGrattageLICS05,DiazcaroGuillermoMiquelValironLICS19,LambdaS1}.
QML~\cite{AltenkirchGrattageLICS05} gives a restricted notion of orthogonality
between terms, and constructs its superpositions only over orthogonal terms.
Lambda-$S_1$~\cite{DiazcaroGuillermoMiquelValironLICS19,LambdaS1} is the
unitary restriction of Lambda-$S$~\cite{LambdaS}, using an
extended notion of orthogonality. This kind of restrictions could be added as
restrictions to the interstitial rules to achieve the same result.

\section{Conclusion}

The link between linear logic and linear algebra has been known for a
long time, in the context of models of linear logic. We have shown in
this paper, that this link also exists at the syntactic level,
provided we consider several proofs of $\one$, one for each scalar, we
add two interstitial rules, and proof reduction rules allowing to
commute these interstitial rules with logical rules, to reduce
commuting cuts.

We also understand better in which way must propositional logic be
extended or restricted, so that its proof language becomes a quantum
programming language.  A possible answer is in four parts: we need to
extend it with interstitial rules, scalars, and the connective
$\odot$, and we need to restrict it by making it linear.  We obtain
this way the ${\mathcal L} \odot^{\mathbb C}$-calculus that addresses
both the question of linearity and, for instance, avoids cloning, and
that of the information-erasure, non-reversibility, and
non-determinism of the measurement.

Future work also includes relating the algebraic notion of tensor product
and the linear logic notion of tensor for vector propositions.

\section*{Acknowledgement}
The authors want to thank Thomas Ehrhard,
Jean-Baptiste Joinet, Jean-Pierre Jouannaud, Dale Miller, Alberto
Naibo, Simon Perdrix, Alex Tsokurov, and Lionel Vaux for useful discussions.

\bibliographystyle{abbrv}
\bibliography{biblio}

\end{document}